\documentclass[12pt]{article}
\usepackage[utf8]{inputenc}
\usepackage[margin=2.5cm]{geometry}
\usepackage{amsmath}
\usepackage{amssymb}
\usepackage{bm}
\usepackage[T1]{fontenc}
\usepackage{graphicx}
\usepackage{enumerate}
\usepackage{natbib}
\usepackage{amsthm}
\usepackage{booktabs}
\usepackage{url} % not crucial - just used below for the URL 
\usepackage{comment}
\usepackage{xcolor}

\newtheorem{prop}{Proposition}[section]
\newtheorem{thm}{Theorem}[section]

\begin{document}
	
	%\bibliographystyle{natbib}
	
	%%%%%%%%%%%%%%%%%%%%%%%%%%%%%%%%%%%%%%%%%%%%%%%%%%%%%%%%%%%%%%%%%%%%%%%%%%%%%%
	
	\title{\bf Design and Structure Dependent Priors for Scale Parameters in Latent Gaussian Models}
	\author{Aldo Gardini\\
		and\\
		Fedele Greco\\
		and\\
		Carlo Trivisano\\
		Department of Statistical Sciences, University of Bologna,\\
		Bologna, 40126, Italy}
	\maketitle
	
	\begin{abstract}
		Many common correlation structures assumed for data can be described through latent Gaussian models. When Bayesian inference is carried out, it is required to set the prior distribution for scale parameters that rules the model components, possibly allowing to incorporate prior information. This task is particularly delicate and many contributions in the literature are devoted to investigating such aspects. We focus on the fact that the scale parameter controls the prior variability of the model component in a complex way since its dispersion is also affected by the correlation structure and the design. To overcome this issue that might confound the prior elicitation step, we propose to let the user specify the marginal prior of a measure of dispersion of the model component, integrating out the scale parameter, the structure and the design. Then, we analytically derive the implied prior for the scale parameter. Results from a simulation study, aimed at showing the behavior of the estimators sampling properties under the proposed prior elicitation strategy, are discussed. Lastly, some real data applications are explored to investigate prior sensitivity and allocation of explained variance among model components.
	\end{abstract}
	
	\noindent%

	\section{Introduction}
	Latent Gaussian Models (LGMs) are a subclass of Generalized Linear Mixed Models where the expected value of a response variable $y$ is connected to a linear predictor $\eta$ via a link function $g(\eta)$. The linear predictor 
\begin{equation}\label{eq:linpred}
	\bm\eta=\bm1\beta_0+\mathbf{X}\boldsymbol{\beta}+\sum_{j=1}^Q\boldsymbol{\nu}_j
\end{equation}
is constituted by a priori independent additive components distributed as Gaussian random variables conditionally on model hyperparameters. The design matrix $\mathbf{X}\in\mathbb{R}^{n\times P}$ is associated to fixed effects $\boldsymbol{\beta}=(\beta_1,\ldots,\beta_P)^\top$, $\beta_0$ is an overall intercept and $\boldsymbol{\nu}_j\in\mathbb{R}^{n}$, $j=1,\ldots,Q$, are $Q$ independent vectors of random effects. Without loss of generality, we consider both covariates and random effects to be centered: this is strongly advised when implementing MCMC algorithms \citep{gelfand-sahu-95} and favors a natural interpretation of the prior specification strategy proposed in this paper.

It is convenient to express random components as the product of a random effect design matrix $\textbf{Z}_j\in\mathbb{R}^{n\times m_j}$, with $m_j\leq n$ and a random vector $\boldsymbol{\gamma}_j\in\mathbb{R}^{m_j}$ which in LGMs follows a Gaussian distribution, i.e. $\boldsymbol{\nu}_j=\textbf{Z}_j\boldsymbol{\gamma}_j$. This allows to encompass models for grouped data, where $\mathbf{Z}_j$ is built as a selection matrix, low-rank models and non-parametric regression where $\mathbf{Z}_j$ is built as a basis matrix. Furthermore, when the random effect is observation-specific, $\mathbf{Z}_j$ corresponds to the identity matrix $\mathbf{I}_n$, as in the case of spatial models for areal data.
In this paper, we consider priors on $\boldsymbol{\gamma}_j$ with \textit{fixed} structure matrix $\mathbf{K}_{\gamma_j}$ reflecting the modeler's prior beliefs on the dependence relationships characterizing the $j$-th random effect. 

Table \ref{tab:intro} summarises the whole model architecture of the LGMs covered in this paper. All model parameters follow a Gaussian distribution:
\begin{equation*}
\beta_p|\sigma^2_{\beta_p}\sim\mathcal{N}_1\left(0,\sigma^2_{\beta_p}\right),\ p=1,\ldots,P,\qquad
	\boldsymbol{\gamma}_j|\sigma_{\gamma_j}^2\sim \mathcal{N}_{m_j}\left(\boldsymbol{0},\sigma^{2}_j\mathbf{K}_{\gamma_j}^{-1}\right),\ j=1,\dots,Q,
\end{equation*}
and the choice of prior distributions for  $\boldsymbol{\sigma}^2=(\sigma_{\beta_1}^2,\ldots,\sigma_{\beta_P}^2,\sigma_{\gamma_1}^2,\ldots,\sigma_{\gamma_Q}^2)$ completes model specification. 
Such hyperparameters act as mere scalers, governing the contribution of each model component to the total prior variability of the linear predictor.

\begin{table}[]
\centering
\begin{tabular}{@{}l|ccc@{}}
\toprule
             & Intercept        & Fixed effects                                 & Random effects                                     \\ \midrule
Coefficients & $\beta_0$            & $\beta_i,\ i=1,\dots,P$                       & $\boldsymbol{\gamma}_j,\ j=1,\dots,Q$                     \\
Design       & $\boldsymbol{1}$ & $\mathbf{x}_i$ & $\mathbf{Z}_j$ \\
Structure    & $1$              & $1$                                   & $\mathbf{K}_{\gamma_j}$                  \\
Scaler       & $10^6$   & $\sigma^2_{\beta_i}$ & $\sigma^2_{\gamma_j}$    \\ \bottomrule
\end{tabular}
\caption{Adopted notation for components of an LGM.}
\label{tab:intro}
\end{table}

In this paper, we develop a prior specification strategy that aims to take into account the effect of design matrices, structure matrices and possible linear constraints in propagating the variability induced by the prior on scalers through the linear predictor. For this reason, we dub the proposed priors as Design and Structure Dependent (DSD) priors. Actually, by introducing dependence of the prior on model architecture, we obtain prior statements on the variability of the linear predictor that are not dependent on design and structure, marginally with respect to hyperparameters. The theoretical developments rely on the theory of Quadratic Forms (QF) in Gaussian variables, requiring specific computational tools \citep{gardini2022}. Starting from a Beta distribution of the $2^{nd}$ kind as a base prior, we reduce the prior elicitation problem to the specification of one parameter that calibrates the model components a priori marginal variability, which is strictly tied to the degree of smoothness in several models. One of the merits of such parameter is to allow for intuitive sensitivity analysis, following the same rationale among models that differ with respect to structure and design. 

The rest of the paper is organized as follows. Section \ref{sec:review-pri} offers a brief review of the approaches suggested in the literature on priors for scale parameters in LGM that are relevant to the theory developed in subsequent sections. 
In Section \ref{sec:sdpri}, we introduce our novel  prior specification strategy on scalers of the random effects distributions; the prior specification for the fixed effects, as will be shown subsequently, can be tackled as a special case. 
A simulation study and real data applications are shown in Sections \ref{sec:simulsec} and \ref{sec:application} respectively, while concluding remarks are offered in Section \ref{sec:conclusion}.

	\section{Hyperpriors on Scale Parameters}\label{sec:review-pri}
	%Due to computational convenience reasons, the hyperprior for variance components were treated in parallel to the choice of the prior distribution of the variance parameter of a Gaussian model, i.e. considering non-informative or inverse-gamma priors.

 As pointed out by \citet{Gelman2006} and \citet{Polson2012}, selecting priors for scale parameters of hierarchical models requires particular attention since such parameters are not directly informed by data. Furthermore, if the scale parameter is set to 0, the model reduces to the one without the random effect (labeled as base model in the literature): \citet{Gelman2006} advises priors with non-null density at 0. The idea of base model is also a cornerstone of the research line on Penalizing-Complexity (PC) priors opened by \citet{Simpson2017}, where non-null density at 0 is suggested in order to favor shrinkage toward the base model.

Besides the distributional assumption presented in Section \ref{sec:B2}, hyperpriors need also care in calibration and, possibly, elicitation of prior knowledge in the model. For these tasks, assumptions on the model components and the available auxiliary information play a crucial role that must be taken into account, as will be discussed in Section \ref{sec:str}.

\subsection{Prior Distribution: the Beta Distribution of the $2^{\text{nd}}$ Kind}\label{sec:B2}

As noted by \citet{Perez2017}, most of the priors proposed for scale parameters are embedded within the Beta distribution of the $2^{\text{nd}}$ kind (B2). The B2 distribution $\text{B2}(b,p,q)$ is a flexible distribution ruled by three parameters, where $b$ is a scale parameter, whereas $p$ and $q$ are shape parameters. 
The density function is:
\begin{equation}
	\label{eq:b2}
	f^{\text{B2}}_{\sigma^2}(s)=\mathcal{K}_{\text{B2}}\;s^{-q-1}\left(1+\frac{b}{s}\right)^{-p-q},
\end{equation}
where $\mathcal{K}_{\text{B2}}=\frac{b^q}{B(p,q)}$ is the normalizing constant and $B(\cdot,\cdot)$ is the beta function. Parameters $p$ and $q$ directly control the behavior of the distribution tails: $p$ determines the behavior in 0 (divergent for $p<1$, finite for $p=1$ and 0 for $p>1$), and $q$ regulates the degree of the polynomial decay of the right tail (the higher $q$, the lighter the tail). A B2 prior on $\sigma^2$ with $p = 0.5$ implies a finite non-null density at 0 for $\sigma$.

Note that the Half-$t$ prior on $\sigma$ with $d$ degrees of freedom and scale $c$ coincides with a $\text{B2}(c^2d,1/2,d/2)$ prior for $\sigma^2$.  Gamma and Inverse Gamma distributions are obtained as limiting cases.

\subsection{The Impact of Structure and Design}\label{sec:str}
An element of criticism in the prior specification emerges whenever a non-diagonal precision matrix is assumed for the random effects distribution, as in this case interpretation of the prior can hardly be disentangled by the structure matrix: Section \ref{sec:distr_V} is devoted to the formalization of such dependence. As an example, in the framework of disease mapping, \citet{bernardinelli1995bayesian} noted that the interpretation of a scale parameter of a spatial random effect is conditional with respect to the neighborhood structure and, therefore, different priors should be assigned to a structured and unstructured component. 
More generally, this consideration can be extended to each case in which an intrinsic Gaussian Markov random field \citep[IGMRF,][]{rue2005gaussian} prior is assumed for a random effect. Indeed, IGMRFs are in general heteroscedastic, being the variance dependent on the structure. For this reason, \citet{Sorbye2014} suggest to specify a scaled hyperprior for $\sigma^2_{\gamma_j}$, accounting for the structure. In particular, they specify a prior on $\sigma^2_{\gamma_j}s^2_{ref,j}$, where $s_{ref,j}$ is the reference standard deviation for the considered IGMRF. Such scaling value is set equal to the geometric mean of the eigenvalues of $\mathbf{K}_{\gamma_j}^-$. The procedure is implemented in the popular INLA software and has been exploited by \citet{Riebler2016} in the context of the popular BYM model \citep{besag1991bayesian}.

In addition, if $\boldsymbol{\nu}_j=\mathbf{Z}_j\boldsymbol{\gamma}_j$, with $\mathbf{Z}_j\neq\mathbf{I}_n$, the design matrix contributes to the prior variability too.
For this reason, \citet{Klein2016} remark that it is not possible to elicit the prior information only considering $\mathbf{K}_{\gamma_j}$ since the dispersion of the whole vector of random effects $\boldsymbol{\nu}_j$ would be ignored. The authors focused on the P-splines regression framework and their elicitation strategy relies on a probabilistic statement about the marginal dispersion of $\boldsymbol{\nu}_j$. As a consequence, the prior on $\sigma_{\gamma_j}^2$ is retrieved integrating out both the random effect variance $\sigma_{\gamma_j}^2\mathbf{K}_{\gamma_j}^{-}$ and the covariate patterns involved in $\mathbf{Z}_j$.
This represents an interesting contribution with respect to the developments proposed in what follows, since the model is calibrated for $\boldsymbol{\nu}_j$, instead of for the vector of sole coefficients $\boldsymbol{\gamma}_j$.

\subsection{Other Approaches: Splitting the Total Variability}\label{sec:split}

Another approach to the problem of specifying the priors and hyperpriors in hierarchical models is based on the idea of splitting the overall model variance among the different components. This can be thought as an extension on the prior on the coefficient of determination $R^2$ for linear models proposed by \citet{gelman_hill_2006}.

We mention two recent interesting proposals within this framework. Narrowing the field to the linear model case, \citet{r2d2} derived a shrinkage prior for the regression coefficient, specifying a marginal prior on the $R^2$ coefficient. An interesting feature of such strategy is that the design is automatically integrated out, deducing the induced priors on the vector of coefficients that are in practice dependent on the design. \citet{yanchenko2021r2d2} extended this procedure to Generalized Linear Mixed Models. With similar spirit, the work by \citet{Fuglstad2020} propose to specify a prior distribution for the total variability and then split it among the distinct model components, accounting for their structure, putting a Dirichlet prior on the splitting nodes.

	\section{Design and Structure Dependent Priors}\label{sec:sdpri}
	
Design and structure dependent priors are introduced by focusing on a random effect $\boldsymbol{\nu}_j=\mathbf{Z}_j\boldsymbol{\gamma}_j$. The specification of the prior on scale parameter is based on its effect on the sampling variance of $\boldsymbol{\nu}_j$, defined as the random variable: 
\begin{equation}
    \label{eq:var_nu}
	V_{\nu_j}=\frac{1}{n-1}\sum_{i=1}^n(\nu_{ij}-\bar{\nu}_j)^2=\frac{\boldsymbol{\nu}_j^\top\mathbf{M}\boldsymbol{\nu}_j}{n-1}=\frac{\boldsymbol{\gamma}_j^\top\mathbf{Z}_j^\top\mathbf{M}\mathbf{Z}_j\boldsymbol{\gamma}_j}{n-1},
\end{equation}
where $\mathbf{M}=\left[\mathbf{I}_n-n^{-1}\boldsymbol{1}_n\boldsymbol{1}_n^T\right]$ is the centering matrix. This is a QF in Gaussian variables, whose distribution depends on both the design matrix $\mathbf{Z}_j$ and the structure matrix $\mathbf{K}_{\gamma_j}$ (for further details, see Section \ref{app:QF} in the Appendix). Note that the structure matrix induced on $\boldsymbol{\nu}_j$ is $\mathbf{K}_{\nu_j}=\left(\mathbf{Z}_j\mathbf{K}_{\gamma_j}^{-1}\mathbf{Z}_j^\top\right)^{-}$.

The prior distribution on $\sigma^2_{\gamma_j}$ is specified in order to control the marginal distribution of $V_{\nu_j}$, i.e. the prior is retrieved integrating out the effect of the structure $\mathbf{K}_{\gamma_j}$ and the design $\mathbf{Z}_{j}$. In this way, we make the prior on $\sigma_{\gamma_j}^2$  dependent in design and structure, with the aim of making the marginal prior on $V_{\nu_j}$ independent on both of them. 
%\textcolor{magenta}{Si recupera? since the role of the prior on $\sigma^2_j$ on the marginal distribution of $V_{\boldsymbol\nu_j}$ cannot be disentangled from the precision matrix $\mathbf{K}_{\nu_j}$}

In Section \ref{sec:distr_V}, the conditional and marginal distributions of $V_{\nu_j}$ are studied. In Section \ref{sec:deriv}, the DSD prior for the case of full-rank precision matrices $\mathbf{K}_{{\gamma}_j}\succ 0$ is derived, whereas the important case of semi-positive definite matrix $\mathbf{K}_{{\gamma}_j}$, which includes the class of IGMRF priors, is considered in Section \ref{sec:IGMRF}. Section \ref{sec:elicit} is devoted to prior elicitation.

%We believe that this is a meaningful quantity to be considered, as it expresses the contribution of each component to the total variability of the linear predictor \eqref{eq:linpred}: indeed, because of prior independence between model components,
%\begin{equation*}
	%\mathbb{E}\left[V_{\eta}\right]=\sum_{j=1}^s\mathbb{E}\left[V_{{\nu}_j}\right].
	%\mathbb{E}\left[V_{\eta}\right]=\sum_{j=1}^Q\mathbb{E}\left[V_{{\nu}_j}\right].
%\end{equation*}

\subsection{Conditional and Marginal Priors on $V_{{\nu}_j}$}\label{sec:distr_V}

The distribution of $V_{{\nu}_j}$ conditional on $\sigma^2_{\gamma_j}$ can be expressed as a linear combination of independent chi-square random variables \citep{box1954,ruben1962probability},
\begin{equation}
	\label{eq:qf}
	V_{{\nu}_j}|\sigma^2_{\gamma_j}\sim \frac{\sigma_{\gamma_j}^2}{n-1}\sum_{k=1}^{n}\lambda_{kj}X_k;\ \ X_k\stackrel{ind}{\sim}  \chi^2_1,\ k=1,\dots,n,
\end{equation}
where $\boldsymbol{\lambda}_j=(\lambda_{1j},\dots, \lambda_{nj})^\top=\text{eig}\left(\mathbf{M}\mathbf{K}_{\nu_j}^{-1}\right)$. Expectation and variance are:
\begin{equation*}
	\mathbb{E}\left[V_{{\nu}_j}|\sigma^2_{\gamma_j}\right]=\frac{\sigma_{\gamma_j}^2}{n-1}\sum_{k=1}^{n}\lambda_{kj} \text{ and } \mathbb{V}\left[V_{{\nu}_j}|\sigma^2_{\gamma_j}\right]=\left(\frac{\sigma_{\gamma_j}^2}{n-1}\right)^2\sum_{k=1}^{n}2\lambda_{kj}^2.
\end{equation*}
We remark that at least one of the eigenvalues is null because of the multiplication by the rank-deficient matrix $\mathbf{M}$. Equation \eqref{eq:qf} reveals that the dependence of the distribution of $V_{\nu_j}|\sigma^2_{\gamma_j}$ on the structure matrix $\mathbf{K}_{{\nu}_j}=\left(\mathbf{Z}_j\mathbf{K}_{\gamma_j}^{-1}\mathbf{Z}_j^\top\right)^{-}$ is captured by the eigenvalues $\boldsymbol{\lambda}_j$: different structure and design matrices lead to different conditional moments of $V_{\nu_j}$. For this reason, a naive prior specification that chooses the same prior distribution for each $\sigma^2_{\gamma_j}$ would imply different marginal distributions of the sampling variances $V_{\nu_j}$, i.e. different contributions of each random component to the a priori sampling variance of the linear predictor.

Following analogous arguments, \cite{Sorbye2014} propose to scale the structure matrices by the geometric mean of their eigenvalues: a similar approach would consist in dividing $\mathbf{K}_{\nu_j}$ by $\frac{\sum_{k=1}^{n}\lambda_{kj}}{n-1}$. This would remove dependence of $\mathbb{E}\left[V_{\nu_j}|\sigma^2_{\gamma_j}\right]$ on $\mathbf{K}_{{\nu}_j}$, obtaining that
$
	\mathbb{E}\left[V_{{\nu}_j}|\sigma^2_{\gamma_j}\right]=\sigma_{\gamma_j}^2,\ \forall j.
$
Nonetheless, scaling $\mathbf{K}_{\nu_j}$ by any constant would not completely remove the dependence of the marginal density  $f_{V_{{\nu}_j}}$ on $\mathbf{K}_{{\nu}_j}$. In other words, if the same prior $f_{\sigma^2}$ is selected for each $\sigma^2_{\gamma_j}$, the marginal distributions
\begin{equation}\label{eq:margV}
	f_{V_{{\nu}_j}}(v)=\int_{0}^{+\infty}f_{V_{{\nu}_j}|\sigma^2_{\gamma_j}}(v,s)f_{\sigma^2}(s)\mathrm{d}s,\quad j=1,\ldots,Q,
\end{equation}
 are all different because $f_{V_{{\nu}_j}|\sigma^2_{\gamma_j}}(\cdot)\neq f_{V_{{\nu}_l}|\sigma^2_{\gamma_l}}(\cdot)$ for each $j\neq l$. In summary:
\begin{equation*}
	f_{\sigma^2_{\gamma_j}}\equiv f_{\sigma^2},\ \forall j \implies f_{V_{{\nu}_j}}\not\equiv f_{V_{{\nu}_l}}, \ \forall l\neq j;
\end{equation*}
confirming that equal priors on scale parameters lead to different priors on random effects sampling variances. 

We stress that assigning different prior importance to each random component is perfectly sensible if it reflects the modeler's prior beliefs: this can be beneficial and even desirable in Bayesian applications. Several efforts have been made in the literature to build intuitive and flexible prior specification strategies that allow to manage the contribution of each random component to the total variability, a recent prominent example being \cite{Fuglstad2020}. 

\subsection{Derivation of the DSD Prior}\label{sec:deriv}

As a starting point, we focus on the effect of the structure matrix by considering a single random effect $\boldsymbol\nu|\sigma^{2}\sim\mathcal{N}_n\left(\mathbf{0},\sigma^{2}\mathbf{K}_{\nu}^{-1}\right)$, with $\mathbf{Z}=\mathbf{I}_n$ and $\mathbf{K}_\nu=\mathbf{K}_\gamma\succ 0$. The structure matrix $\mathbf{K}_{\nu}$ encodes the conditional relationships between elements of $\boldsymbol\nu$ and it depends on the application at hand.

%(\textit{toglierei o sposterei : to give a couple of examples, if $\mathbf{K}_{\nu}$ is built to model spatial dependence between areas in a discrete spatial domain, its eigenvalues will depend both on $n$ and on the neighbouring structure of the map. On the other hand, if $\mathbf{K}_{\nu}$ represents the structure matrix of a random walk (RW) process, its eigenvalues depend on $n$ and the skewness of the spectrum rises with $n$, with a speed depending on the order of the process.})
A notable special case, that plays a prominent role in the proposed prior specification strategy, is represented by i.i.d. random effects obtained by setting $\mathbf{K}_{{\nu}}=\mathbf{I}_n$. This is an ideal benchmark for managing prior specification on scale parameters because of its simplicity and interpretability: in the i.i.d. model, independently on $n$, $\sigma^2$ corresponds to both the expected sampling variance of $\boldsymbol\nu$ and to the variance of each component $\nu_i$. The sampling variance of an i.i.d. random effect, denoted as $\mathcal{V}$ in what follows, is distributed as a scaled chi-square with $n-1$ degrees of freedom conditionally on $\sigma^2$. Equivalently, posing $\alpha=\frac{n-1}{2}$ and $\beta=\frac{n-1}{2}$, one obtains $\mathcal{V}|\sigma^2\sim\text{Gamma}\left(\alpha,\frac{\beta}{\sigma^2}\right),$
with $\mathbb{E}\left[\mathcal{V}|\sigma^2\right]=\mathbb{V}\left[\nu_i|\sigma^2\right]=\sigma^2$.
Given the generality of the distribution, we opt for a B2 prior (introduced in Section \ref{sec:B2}) on $\sigma^2$ in the i.i.d. case. Consequently, the marginal distribution of $\mathcal{V}$ is a mixture of a gamma and a B2 distributions whose density function is given in Proposition \ref{prop:2f0}. Since this marginal density can be expressed in terms of the $_2F_0$ hypergeometric function \citep{NIST}, we dub it $_2\mathcal{F}_0$-distribution.

\begin{prop}[$_2\mathcal{F}_0$-distribution]\label{prop:2f0}
Let $X$ and $Y$ be two random variables such that
\begin{equation*}
	X|y\sim \mathrm{Gamma}\left(\alpha,\frac{\beta}{y}\right),\quad Y\sim \mathrm{B2}(b,p,q).
\end{equation*}
Then, $X\sim{_2}\mathcal{F}_0(\alpha,\beta/b,p,q)$ where $\alpha$, $p$ and $q$ are shape parameters and the scale parameter $\beta/b$ is the ratio of the scale parameters of the mixed distributions. 
The density of $X$ is:
\begin{align}\label{eq:MG_f}
	&f_X(x)=\left(\frac{b}{\beta}\right)^q\frac{\Gamma\left(1+\alpha-p\right)}{ \Gamma(p)B(\alpha,q)}x^{-q-1}\ _2F_0\left(\alpha+q,p+q;-;-\frac{b}{x\beta}\right).
\end{align}

\end{prop}

\begin{proof}
The problem of proving this result can be tackled from different perspectives. Here, the proposed way relies on the inversion of a Mellin transform. The starting point of the proof is standard:
\begin{equation*}
	f_X(x)=\int_0^{+\infty}f_{X|Y}(x,y)f_Y(y)\mathrm{d}y.
\end{equation*}
Knowing that $f_{X|Y}(x,y)=y^{-1}f_{X^*|Y}(x/y)$, where $X^*=y^{-1}X$:
\begin{equation}\label{eq:mult}
	f_X(x)=\int_0^{+\infty}y^{-1}f_{X^*|Y}(x/y)f_Y(y)\mathrm{d}y,
\end{equation}
where $f_{X^*|Y}(x/y)$ is the density function of a gamma distribution with parameters $(\alpha, \beta)$ evaluated in $x/y$. Given the form of the integral, relationship \eqref{eq:mellin_rel} can be exploited, after the Mellin transform of the Gamma distribution and the B2 distribution are retrieved, with the related strips of analyticity:
\begin{align}
	&\widehat{f}_{X^*|Y}(z)=\left(\frac{1}{\beta}\right)^{z-1}\frac{\Gamma\left(\alpha+z-1\right)}{\Gamma(\alpha)},\quad \Re(z)>1-\alpha;\label{eq:mellin_gamma}\\
	&\widehat{f}_{Y}(z)=b^{z-1}\frac{\Gamma(p+z-1)\Gamma(q-z+1)}{\Gamma(p)\Gamma(q)},\quad 1-p<\Re(z)<1+q.\label{eq:mellin_B2}
\end{align}
Plugging these expressions into formula \eqref{eq:mellin_rel}, leads us to the Mellin transform  $\widehat{f}_X(z)$, that has the following strip of analyticity:
$$
\max\left(1-p,1-\alpha\right)<\Re(z)<1+q,
$$
obtained intersecting those of \eqref{eq:mellin_gamma} and \eqref{eq:mellin_B2}.
Lastly, the density function $f_X(x)$ is obtained by inverting the Mellin transform, exploiting the generic formula \eqref{eq:mellin_inv}:
\begin{align*}
	f_X(x)&=\frac{1}{2\pi i}\int_{h-i\infty}^{h+i\infty} x^{-z}\left(\frac{b}{\beta}\right)^{z-1}\frac{\Gamma\left(\alpha+z-1\right)\Gamma(p+z-1)\Gamma(q-z+1)}{\Gamma(\alpha)\Gamma(p)\Gamma(q)}\mathrm{d}z\\
	&=\frac{\beta}{2b \Gamma(\alpha)\Gamma(p)\Gamma(q)\pi i}\int_{h-i\infty}^{h+i\infty}\left(\frac{b}{x\beta}\right)^{z}\Gamma\left(\alpha+z-1\right)\Gamma(p+z-1)\Gamma(q-z+1)\mathrm{d}z\\
	&=\frac{\beta\Gamma\left(\alpha+q\right)\Gamma\left(1+\alpha-p\right)}{b \Gamma(\alpha)\Gamma(p)\Gamma(q)}\left(\frac{x\beta}{b}\right)^{\alpha-1}U\left(\alpha+q,1+\alpha-p,\frac{x\beta}{b}\right),
\end{align*}
where the last step directly follows from the integral representation of Kummer's $U$ hypergeometric function in terms of Mellin-Barnes integrals \citep[][equation 13.4.17]{NIST}.
\end{proof}
	
Thus, for the i.i.d. model, the marginal distribution of $\mathcal{V}$ is
\begin{equation}
	\label{eq:margiid}
	\mathcal{V}\sim{_2}\mathcal{F}_0(\alpha,\beta/b,p,q).
\end{equation}
However, if $\mathbf{K}_\nu\neq\mathbf{I}_n$, the sampling variance of the random effect conditioned on $\sigma^2$ is a linear combination of chi-squared random variables (see equation \eqref{eq:qf}), hence a $\text{B2}(b,p,q)$ prior on $\sigma^2$ would result in a different marginal distribution of $V_\nu$, which depends on the application at hand: we argue that this dependence is undesirable since it impedes coherence of prior statements among different models. To remove such dependence, we aim to obtain a DSD prior on $\sigma^2$, denoted as $f_{\sigma^2}^{DSD}$, that, when mixed with the distribution of $V_\nu|\sigma^2$, delivers the density of the benchmark $_2\mathcal{F}_0(\alpha,\beta/b,p,q)$ distribution. Technically, we seek the prior density which solves the integral equation:
\begin{equation*}
	f_{\mathcal{V}}(v)=\int_{0}^{+\infty}f_{V|\sigma^2}(v,s)f^{DSD}_{\sigma^2}(s)\mathrm{d}s.
\end{equation*}
The merit of such distribution would be to have the same interpretation, independently on $\mathbf{K}_{\nu}$, in terms of marginal variance, which will always coincide with that of the i.i.d. model.

For the sake of simplicity, we develop DSD priors working on the following approximation of $V|\sigma^2$:
\begin{equation}\label{eq:approx_qf}
	V|\sigma^2\stackrel{a}{\sim} \text{Gamma}\left(\tilde{\alpha},\frac{\tilde{\beta}}{\sigma^2}\right),
\end{equation}
where $\stackrel{a}{\sim}$ indicates a random variable \textit{approximately distributed as}. This approximation has been proposed by \cite{box1954} in order to match the first two moments of $V|\sigma^2$, and parameters are
\begin{equation}\label{eq:tildes}
\tilde{\alpha}=\frac{\left(\sum_{i=1}^{n}\lambda_i\right)^2}{2\sum_{i=1}^{n}\lambda^2_i},\qquad\tilde{\beta}=\frac{n-1}{2}\frac{\sum_{i=1}^{n}\lambda_i}{\sum_{i=1}^{n}\lambda^2_i}.
\end{equation}
Since $\tilde{\alpha}$ contains the ratio of the square of a sum over a sum of squares, $\tilde{\alpha}\geq 1/2$. Given that  $V|\sigma^2$ approximately follows a Gamma distribution, Proposition \ref{prop:2f0} implies that, if $\sigma^2\sim\text{B2}(b,p,q)$, the approximate marginal distribution of $V$ is
\begin{equation}
	\label{eq:margVtil}
	V \stackrel{a}{\sim}{_2}\mathcal{F}_0(\tilde{\alpha},\tilde{\beta}/b,p,q).
\end{equation}
Comparison between \eqref{eq:margiid} and \eqref{eq:margVtil} clearly highlights that the difference between these marginal distributions is due to the eigenvalues of $\mathbf{MK}_\nu^{-1}$. The following theorem constitute the main result of the paper, stating the density of the DSD prior, i.e. a prior that delivers the marginal density \eqref{eq:margiid} independently on $\mathbf{K}_\nu$. The prior density turns out to be expressed in terms of the $_2F_1$ hypergeometric function \citep{NIST}. 

\begin{thm}[DSD prior] \label{thm:balpri}
Let $\boldsymbol{\nu}|\sigma^2\sim\mathcal{N}_n\left(\mathbf{0},\sigma^{2}\mathbf{K}^{-1}_\nu\right)$. The DSD prior $f_{\sigma^2}^{DSD}$, i.e. the prior that solves the integral equation
\begin{equation}\label{eq:inteqth}
	\int_{0}^{+\infty}f_{V|\sigma^2}(v,s)f^{DSD}_{\sigma^2}(s)\mathrm{d}s=f_\mathcal{V}(v),
\end{equation}
where $\mathcal{V}\sim{_2}\mathcal{F}_0(\alpha,\beta/b,p,q)$ and $V|\sigma^2\stackrel{a}{\sim} \text{Gamma}\left(\tilde{\alpha},\frac{\tilde{\beta}}{\sigma^2}\right)$, has density
\begin{equation}\label{eq:dens_balpri}
	f^{DSD}_{\sigma^2}(s)=\mathcal{K}_{\text{DSD}}\;s^{-q-1} {_2F_1}\left(q+\alpha,q+p;q+\tilde{\alpha};-\frac{b\tilde{\beta}}{s\beta}\right),
\end{equation}
with
$$\mathcal{K}_{\text{DSD}}=\left(\frac{b\tilde{\beta}}{\beta}\right)^q\frac{1}{B(p,q)}\frac{\Gamma(\tilde{\alpha})}{\Gamma(q+\tilde{\alpha})}\frac{\Gamma(q+\alpha)}{\Gamma(\alpha)},$$
provided that $p\leq\tilde{\alpha}$.
\end{thm}

\begin{proof}
Equation \eqref{eq:inteqth} represents a Fredholm integral equation of the first kind \citep[Chapter~10]{polyanin2008handbook}. Defining the random variable $V^*=s^{-1}V$ one obtains:
\begin{equation*}
	\int_{0}^{+\infty}s^{-1}f_{V^*|\sigma^2}(v/s)f^{DSD}_{\sigma^2}(s)\mathrm{d}s=f_\mathcal{V}(v),
\end{equation*}
where $f_{V^*|\sigma^2}(v/s)$ is the density function of a gamma random variable with parameters $(\tilde{\alpha}, \tilde{\beta})$ evaluated at $v/s$. In the framework of integral equations, the latter function is also called kernel: since it is expressed as a function of the variables $(v,s)$ through their ratio, the solution of the equation can be retrieved using the Mellin transform tool (see Section \ref{sec:mellin} in the Appendix). Due to relation \eqref{eq:mellin_rel}, the following equality holds:
\begin{equation}
	\widehat{f}_\mathcal{V}(z)=\widehat{f}_{V^*|Y}(z)\widehat{f}^{DSD}_{\sigma^2}(z),
\end{equation}
where $\widehat{f}^{DSD}_{\sigma^2}(z)$ is the Mellin transform of the structure dependent prior. Recalling that $\mathcal{V}\sim {_2}\mathcal{F}_0\left(\alpha, \beta/b,p,q\right)$ and $V|\sigma^2\sim \text{Gamma}(\tilde{\alpha},\tilde{\beta})$, their Mellin transforms can be recovered from equations \eqref{eq:mellin_gamma} and \eqref{eq:mellin_B2} reported for the proof of Proposition \ref{prop:2f0}, obtaining:
\begin{equation*}
	\widehat{f}^{DSD}_{\sigma^2}(z)=\left(\frac{b\tilde{\beta}}{\beta}\right)^{z-1}\frac{\Gamma(\tilde{\alpha})}{\Gamma(\alpha)\Gamma(p)\Gamma(q)}\cdot\frac{\Gamma\left(\alpha+z-1\right)\Gamma(p+z-1)\Gamma(q-z+1)}{\Gamma(\tilde{\alpha}+z-1)},
\end{equation*}
with strip of analyticity $\max\left(1-p,1-\alpha,1-\tilde{\alpha}\right)<\Re(z)<1+q$. Hence, $f_{\sigma^2}^{DSD}$ can be retrieved inverting the Mellin transform:
\begin{equation*}
	f_{\sigma^2}^{DSD}(s)=\frac{1}{2\pi i}\cdot \frac{\beta\Gamma(\tilde{\alpha})}{b\tilde{\beta}\Gamma(\alpha)\Gamma(p)\Gamma(q)}\int_{h-i\infty}^{h+i\infty} \left(\frac{b\tilde{\beta}}{\beta s}\right)^{z}\frac{\Gamma\left(\alpha+z-1\right)\Gamma(p+z-1)\Gamma(q-z+1)}{\Gamma(\tilde{\alpha}+z-1)}\mathrm{d}z.
\end{equation*}
Because of the property of the Mellin transform, convergence of the integral is guaranteed if $h$ is in the strip of analyticity. The integral can be expressed in terms of a $_2F_1(a,b;c;x)$ hypergeometric function \citep[][equation 16.5.1]{NIST}, delivering:
\begin{equation*}
	f^{DSD}_{\sigma^2}(s)=\frac{\beta\Gamma(\tilde{\alpha})}{b\tilde{\beta}\Gamma(\alpha)\Gamma(p)\Gamma(q)}\frac{\Gamma(q+\alpha)\Gamma(q+p)}{\Gamma(q+\tilde{\alpha})}\left(\frac{s\beta}{b\tilde{\beta}}\right)^{-q-1} {_2F_1}\left(q+\alpha,q+p;q+\tilde{\alpha};-\frac{b\tilde{\beta}}{s\alpha}\right).
\end{equation*}
Once the solution of the integral equation is found, assessing if it constitutes a proper density function is necessary. To this aim, it is required to check if $\int_0^{+\infty}f^{DSD}_{\sigma^2}(s)\mathrm{d}s=1$ and $f^{DSD}_{\sigma^2}(s)>0,\ s\in\mathbb{R}^+$. The first property can be proved showing that:
\begin{equation*}
	\int_0^{+\infty}\left(\frac{s\beta}{b\tilde{\beta}}\right)^{-q-1} {_2F_1}\left(q+\alpha,q+p;q+\tilde{\alpha};-\frac{b\tilde{\beta}}{s\alpha}\right)\mathrm{d}{s}=\left(\frac{\beta}{b\tilde{\beta}}\right)^{-1}\frac{\Gamma(\alpha)\Gamma(p)\Gamma(q)\Gamma(q+\tilde{\alpha})}{\Gamma(\tilde{\alpha})\Gamma(q+\alpha)\Gamma(q+p)},
\end{equation*}
applying result 7.511 of \citet{gradshteyn2007}.
On the other hand, to verify the positivity of the function, the following integral representation of the hypergeometric function can be exploited \citep[Equation~9.111]{gradshteyn2007}:
\begin{equation*}
	_2F_1\left(a,b;c;x\right)=\frac{\Gamma(c)}{\Gamma(b)\Gamma(c-b)}\int_0^1t^{b-1}(1-t)^{c-b-1}(1-tx)^{-a}\mathrm{d}t,\ \Re(c)>\Re(b)> 0.
\end{equation*}
Provided that the representation holds, it is possible to note that if $x<0$ the result is positive. In our case, if $\tilde{\alpha}>p$ the function $f^{DSD}_{\sigma^2}(s)$ is positive and, hence, it is a probability density function.
 \end{proof}

Comparison between \eqref{eq:b2} and \eqref{eq:dens_balpri} clearly shows how DSD priors modify the base B2 distribution. Indeed, the kernels of these distributions differ with respect to the last factor: they both lay in the interval $(0,1)$ and are increasing functions of $s$. The couples $(\alpha, \beta)$ and $(\widetilde{\alpha}, \widetilde{\beta})$ determine the growth rate of such factors leading to different shapes of DSD prior with respect to the base B2 prior: an example of such modification is reported in Figure \ref{prior_sim}.  
As a notable special case of Theorem~\ref{thm:balpri}, due to the properties of the $_2F_1$ function, if $V=\mathcal{V}$, i.e. $\tilde{\alpha}=\alpha$ and $\tilde{\beta}=\beta$, the prior reduces to $\sigma^2\sim B2(b,p,q)$. Furthermore, if $p=\tilde{\alpha}$ the DSD prior is $\sigma^2\sim \text{B2}\left(\frac{b\tilde{\beta}}{\beta},\alpha, q\right)$.

The dependence of the prior on $\mathbf{K}_\nu$ is captured by $\widetilde{\alpha}$ and $\widetilde{\beta}$ as defined in \eqref{eq:tildes}: all the generalizations proposed in what follows are based on the fact that, to obtain the DSD prior, one must consider the appropriate eigenvalues vector, i.e. those that express the weights of the associated QF $V_\nu|\sigma^2$. 

Thus, taking account of the effect of the design is quite a straightforward task. Indeed, when $\boldsymbol{\nu}=\mathbf{Z}\boldsymbol\gamma$ with $\mathbf{Z}\in\mathbb{R}^{n\times m}$ and $\boldsymbol{\gamma}|\sigma^2\sim\mathcal{N}_m\left(\boldsymbol{0},\sigma^{2}\mathbf{K}_{\gamma}^{-1}\right)$, with  $\mathbf{K}_\gamma \succ 0$, the appropriate eigenvalues are those of the matrix $\mathbf{M}\mathbf{Z}\mathbf{K}_{\boldsymbol{\gamma}}^{-1}\mathbf{Z}^\top$ whose non-null elements coincide with non-null eigenvalues of $\mathbf{Z}^\top\mathbf{M}\mathbf{Z}\mathbf{K}_{\boldsymbol{\gamma}}^{-1}$. 

The very special case of $m=1$ is useful to shed a light on the rationale behind DSD priors. Indeed, this corresponds to managing a fixed effect, and, for this reason, we switch the notation to
$\boldsymbol{\nu}=\mathbf{x}\beta$ where $\mathbf{x}\in\mathbb{R}^{n\times 1}$ and
$\beta|\sigma^{2}\sim\mathcal{N}\left(0,\sigma^{2}\right)$.
The sampling variance
\begin{equation*}
	V_{\boldsymbol{\nu}}=\frac{\boldsymbol{\nu}^\top\mathbf{M}\boldsymbol{\nu}}{n-1}=\frac{\mathbf{x}^\top\mathbf{M}\mathbf{x}}{n-1}\beta^2=\frac{\sum_{i=1}^n(x_i-\bar x)^2}{n-1}\beta^2=\beta^2s^2_x,
\end{equation*}
has conditional distribution $V_{\boldsymbol{\nu}}|\sigma^2\sim \text{Gamma}\left(\frac{1}{2},\frac{1}{2\sigma^2s^2_x}\right)$
since the only non-null weight of the QF is $s^2_x$.  If $p=\frac{1}{2}$, as will be justified in Section \ref{sec:elicit}, the DSD prior on $\sigma^2$ turns out to be
$
	\sigma^2\sim\text{B2}\left(\frac{b}{s^2_x},\alpha,q\right)
$.
As a consequence, in a model with $P$ fixed effects, the same prior distribution on all standardized coefficients arises.

\subsection{The Case of Intrinsic priors}\label{sec:IGMRF}

IGMRF priors, that are widely employed in spatial and spatio-temporal modelling of areal data and in low-rank models, are improper priors with sparse rank-deficient precision matrix. Let
$\text{Rank}\left(\mathbf{K_\nu}\right)=n-\kappa$: the rank-deficiency $\kappa$ is defined as the order of the IGMRF by \cite{rue2005gaussian}. Again, we start by assuming $\mathbf{Z}=\mathbf{I}_n$.
The developments in what follows are based on the spectral decomposition $\mathbf{K_\nu}=\mathbf{U}\boldsymbol{\Lambda}\mathbf{U}^\top
=\mathbf{U}_+\boldsymbol{\Lambda}_+\mathbf{U}_+^\top$,
where $\boldsymbol{\Lambda}_{+}\in\mathbb{R}^{(n-\kappa)\times (n-\kappa)}$ is a diagonal matrix with diagonal entries corresponding to the non-null eigenvalues of $\mathbf{K_\nu}$ and $\mathbf{U}_+$ spans the column space of $\mathbf{K_\nu}$; in addition, $\mathbf{U}_0$ spans the null space. 

From a probabilistic viewpoint, an IGMRF of order $\kappa$ embeds a proper distribution on the $(n-\kappa)$-dimensional column space of $\mathbf{K}_\nu$ describing deviations from the $\kappa$-dimensional null space, i.e. the implicit systematic part of the model. This important feature of IGMRF priors is highlighted in \cite{rue2005gaussian}, Section 3.4.1, where $\boldsymbol{\nu}$ is decomposed as:
\begin{equation}
\label{eq:dec_igmrf}
\begin{aligned}
	\boldsymbol{\nu}=\text{trend}(\boldsymbol{\nu})+\text{residuals}(\boldsymbol{\nu})
	=\mathbf{U}_0\mathbf{U}_0^\top\boldsymbol{\nu}+\mathbf{U}_{+}\mathbf{U}_+^\top\boldsymbol{\nu}
	=\mathbf{U}_0\boldsymbol{\nu}_0+\mathbf{U}_{+}\boldsymbol{\nu}_{+},
\end{aligned}	
\end{equation}
and $\boldsymbol{\nu}_+|\sigma^2\sim\mathcal{N}_{n-\kappa}\left(\boldsymbol{0},\sigma^{2}\boldsymbol{\Lambda}_+^{-1}\right)$. This decomposition has also been used in \citet{goicoa2018} to study the need for linear constraints in spatio-temporal disease mapping and in \citet{Klein2016} for obtaining their scale dependent priors.

Identifiability of $\boldsymbol{\nu}$ can be ensured by adopting the linear constraint $\mathbf{U}_0^\top\boldsymbol{\nu}=\boldsymbol 0$, as suggested in \cite{held-env-2011}. When implementing LGMs with IGMRF priors, we always include $\boldsymbol{U}_0$ in the linear predictor, if its columns are not linearly dependent on other model components. The prior on $\boldsymbol{\nu}_0$ is specified as in Section \ref{sec:deriv}. To give some examples, if a random walk (RW) prior of order 1 is employed, the null space is spanned by the unit vector and hence it is already included in the intercept term, while the null space of a RW of order 2 is spanned by a first order polynomial, requiring the inclusion of a linear trend.

The sampling variance of a \textit{constrained} IGMRF random effect is
\begin{equation}
	V_{\boldsymbol{\nu}|\mathbf{U}_0^\top\boldsymbol{\nu}=\boldsymbol 0}=\frac{\boldsymbol{\nu}_+^\top\mathbf{U}_+^\top\mathbf{M}\mathbf{U}_+\boldsymbol{\nu}_+}{n-1},
\end{equation}
with  $\boldsymbol{\lambda}=\text{eigen}\left(\mathbf{U}_{+}^\top\mathbf{M}\mathbf{U}_{+}\boldsymbol{\Lambda}_{+}^{-1}\right)=\text{eigen}\left(\mathbf{M}\mathbf{K}_{\nu}^{-}\right)$ being the eigenvalues to be taken into account for computing $\widetilde{\alpha}$ and $\widetilde{\beta}$. 

This can be directly extended to the case $\boldsymbol{\nu}=\mathbf{Z}\boldsymbol{\gamma}$, $\mathbf{Z}\in\mathbb{R}^{n\times m}$, $\boldsymbol{\gamma}|\sigma^2\sim\mathcal{N}_m\left(\mathbf{0},\sigma^{2}\mathbf{K}_\gamma^{-}\right)$. Given the decomposition $\boldsymbol{\gamma}=\mathbf{U}_{\gamma 0}\boldsymbol{\gamma}_{0}+\mathbf{U}_{\gamma+}\boldsymbol{\gamma}_{+}$
 obtained by applying \eqref{eq:dec_igmrf} to $\boldsymbol\gamma$, one gets
$
	\boldsymbol{\nu}=\mathbf{ZU}_{\gamma0}\boldsymbol{\gamma}_{0}+\mathbf{ZU}_{\gamma+}\boldsymbol{\gamma}_{+}.
$
Introducing the linear constraint $\mathbf{U}_{\gamma0}^\top\mathbf{Z}^\top\boldsymbol{\nu}=\mathbf{U}_{\gamma0}^\top\mathbf{Z}^\top\mathbf{Z}\boldsymbol{\gamma}=\boldsymbol{0}$, the sampling variance is
\begin{equation}
	V_{\boldsymbol{\nu}|\mathbf{U}_{\gamma0}^\top\mathbf{Z}^\top\boldsymbol{\nu}=\boldsymbol{0}}=\frac{\boldsymbol{\gamma}_+^\top\mathbf{U}_{\gamma+}^\top\mathbf{Z}^\top\mathbf{MZ}\mathbf{U}_{\gamma+}\boldsymbol{\gamma}_+}{n-1},
\end{equation}
i.e. a QF whose weights are the non-null eigenvalues of the semi-positive definite matrix $\mathbf{U}_{\gamma+}^\top\mathbf{Z}^\top\mathbf{M}\mathbf{Z}\mathbf{U}_{\gamma+}\boldsymbol{\Lambda}_{\gamma+}^{-1}$ that coincide with the non-null eigenvalues of $\mathbf{Z}^\top\mathbf{M}\mathbf{Z}\mathbf{K}_{\gamma}^-$.

It is worth noting that linear constraints on random effects envisioned in this section guarantee property of the posterior distribution: indeed, linear constraints address partial improperty of IGMRF priors, by removing the improper prior on the implicit systematic part of the model. For a comprehensive discussion on this topic in the context of structured additive distributional regression, see \cite{Klein2016}.

\subsection{Prior Elicitation}\label{sec:elicit}

The aim of this section is to deliver a simple and intuitive prior elicitation strategy that can be adopted for every LGM whose architecture falls within Table \ref{tab:intro}. This is favored by DSD priors thanks to their ability to homogenize the prior interpretation in terms of the i.i.d. model, independently on the considered LGM. The B2 distribution serves as a base prior to be adapted to the considered model by solving integral equation \eqref{eq:inteqth}: this guarantees the same marginal distribution of the sampling variances $f_{V_{\nu_j}}\equiv f_\mathcal{V},\ \forall j$. As a starting point, we suggest default values of the shape parameters $p$ and $q$. As regards the scale parameter $b$, we find it convenient to set it by considering observed data variability. 

We set $p=1/2$, leading to the specification of an Half-$t$ distribution on $\sigma$, because this guarantees a finite non-null density at 0, as suggested in \citet{Perez2017}. Moreover, this choice preserves the conditions for the existence of DSD priors stated in Theorem \ref{thm:balpri}, given that $\tilde{\alpha}>1/2$ holds for every LGM. 
We suggest $q=1.5$ as a default choice: this delivers an Half-$t$ with 3 degrees of freedom on $\sigma$. However, we show in subsequent simulations and applications that posterior inference shows small sensitivity to $q$. 

The scale parameter $b$ is the crucial quantity to be specified, because of its impact on the amount of shrinkage/smoothness of the posterior estimates. Linking this parameter to the variability shown by observed data allows automatic scaling of the prior to the considered application. Such scaling is achieved by means of a probability statement on the variability of the $j$-th model component, in the same spirit of \citet{wakefield2007disease,Klein2016} and \citet{Simpson2017}, among the others. 
In particular, we control the probability $\pi_0$ that the marginal sampling variance of a model component $\boldsymbol\nu_j$ is lower than a value $c$:
\begin{equation}
	\label{eq:pi0}
	\mathbb{P}\left[V_{\nu_j}\leq c \right]=\pi_0,\quad \forall j.
\end{equation}
When the likelihood is Gaussian, we set $c=s_y^2$, where $s_y^2$ denotes the sample variance of the response. Hence, the data variance is the $\pi_0$-th quantile of the marginal distribution of $V_{\nu_j},\;\forall j$: the higher $\pi_0$ the lower the prior variability, with consequent heavier shrinkage.
In the case of non-Gaussian likelihood, we find it appealing to resort to the concept of pseudo-variance \citep{gelmanrubin2013}. For an overview of GLMs pseudo-variances under canonical link functions, see Table 1 in \cite{piironen2017}.
Once $c$ is fixed, $b$ is retrieved by solving via Monte Carlo simulation the equation $\mathbb{P}\left[b V_{\nu_j}^*\leq c \right]=\pi_0$, where $V_{\nu_j}^*=V_{\nu_j}/b\sim {_2}\mathcal{F}_0(\alpha,\beta,p,q)$.

Summarizing, the base prior $\text{B2}(b=f(y,\pi_0), 0.5, 1.5)$ is suggested as a default choice, so that
prior specification of a given LGM reduces to set the parameter $\pi_0$, while DSD priors filter out the effect of design and structure matrices as well as possible linear constraints that hamper comparability of priors in different models. In Sections \ref{sec:simulsec} and \ref{sec:application} we show that $\pi_0$ is a decisive quantity in terms of impact on posterior inference: the study of the variation of posterior inference with respect to $\pi_0$ is able to give insights on sensitivity to prior specification.

%\textit{spostare} When the Gaussian model is assumed for the data, a prior for the residual variance must be set too. In line with the indications provided within the \texttt{rstanarm} package \citep{rstanarm}, we decided to specify an exponential distribution on the residual standard deviation, with rate parameter $1/s_y^2$.

\subsection{Decomposition of the Linear Predictor Sampling Variance}

In order to provide a further illustration of the rationale behind DSD priors, we discuss the \textit{a priori} decomposition of the linear predictor sampling variance marginally with respect to scalers $\boldsymbol{\sigma}^2$. Given the equivalence between specifying priors on scalers for fixed and random effects noticed at the end of Section \ref{sec:deriv}, the linear predictor \eqref{eq:linpred} can be expressed as
\begin{equation*}
\boldsymbol{\eta}=\mathbf{1}\beta_0+\sum_{j=1}^{P+Q}\boldsymbol{\nu}_j,
\end{equation*}
where model components refer to fixed effects for $j\leq P$ and to random effects for $P<j\leq P+Q$. The sampling variance of the linear predictor, known as regression variance in the context of linear regression models, is the random variable:
\begin{equation*}
V_\eta=\sum_{j=1}^{P+Q}V_{\nu_j}+2\sum_{j=1}^{P+Q}\sum_{k<j}C_{\nu_j,\nu_k},
\end{equation*}
where terms $C_{\nu_j,\nu_k}$ refer to the bilinear form
$
C_{\nu_j,\nu_k}=\boldsymbol{\nu_j}^\top\mathbf{M}\boldsymbol{\nu_k}/(n-1)$.
Since variance components are a priori independent, $\mathbb{E}\left[C_{\nu_j,\nu_k}\right]=0,\;\forall j\neq k$, and DSD prior guarantee that, $\mathbb{E}\left[V_\eta\right]=(P+Q)\mathbb{E}\left[\mathcal{V}\right]$. Concerning the posterior distributions of $V_{\nu_j}|\mathbf{y}$ and $C_{\nu_j,\nu_k}|\mathbf{y}$, in general, $\mathbb{E}\left[C_{\nu_j,\nu_k}|\mathbf{y}\right]\neq 0$ since the likelihood function induces dependence among model components: from an applied point of view, posterior analysis of variance/covariance decomposition can give insights about the explanatory power of each model component without ignoring their underlying relationships. 
	\section{Simulation Exercise}\label{sec:simulsec}
	The goal of the simulation exercise is to point out the stability of the DSD priors with respect to changes in the design. As data generating process we consider
$$
y_i\sim\mathcal{N}\left(g(x_i)=5+\sin(\pi x_i),\sigma^2_y\right),\quad i=1,\dots,50;
$$
where $\mathbf{x}$ is an equally spaced set of values in the interval $[-1;1]$. 
A semi-parametric regression model implementing Bayesian P-splines \citep{lang2004bayesian} is specified:
$$
\boldsymbol{\eta}=\boldsymbol{1}\mu+ \mathbf{x}\beta+\mathbf{Z}\boldsymbol{\gamma},
$$
where the design matrix $\mathbf{Z}$ is a basis matrix of cubic B-splines bases over $m$ equally spaced knots. For the spline coefficients $\boldsymbol{\gamma}$, a second-order RW prior is imposed, leading to a precision matrix $\mathbf{K}_\gamma$ with rank $m-2$.
The performances of point and interval estimators of $g(x_i)$ are monitored by computing posterior means and $90\%$ credible intervals. 

\begin{figure}[]
\centering
\includegraphics[width=\linewidth]{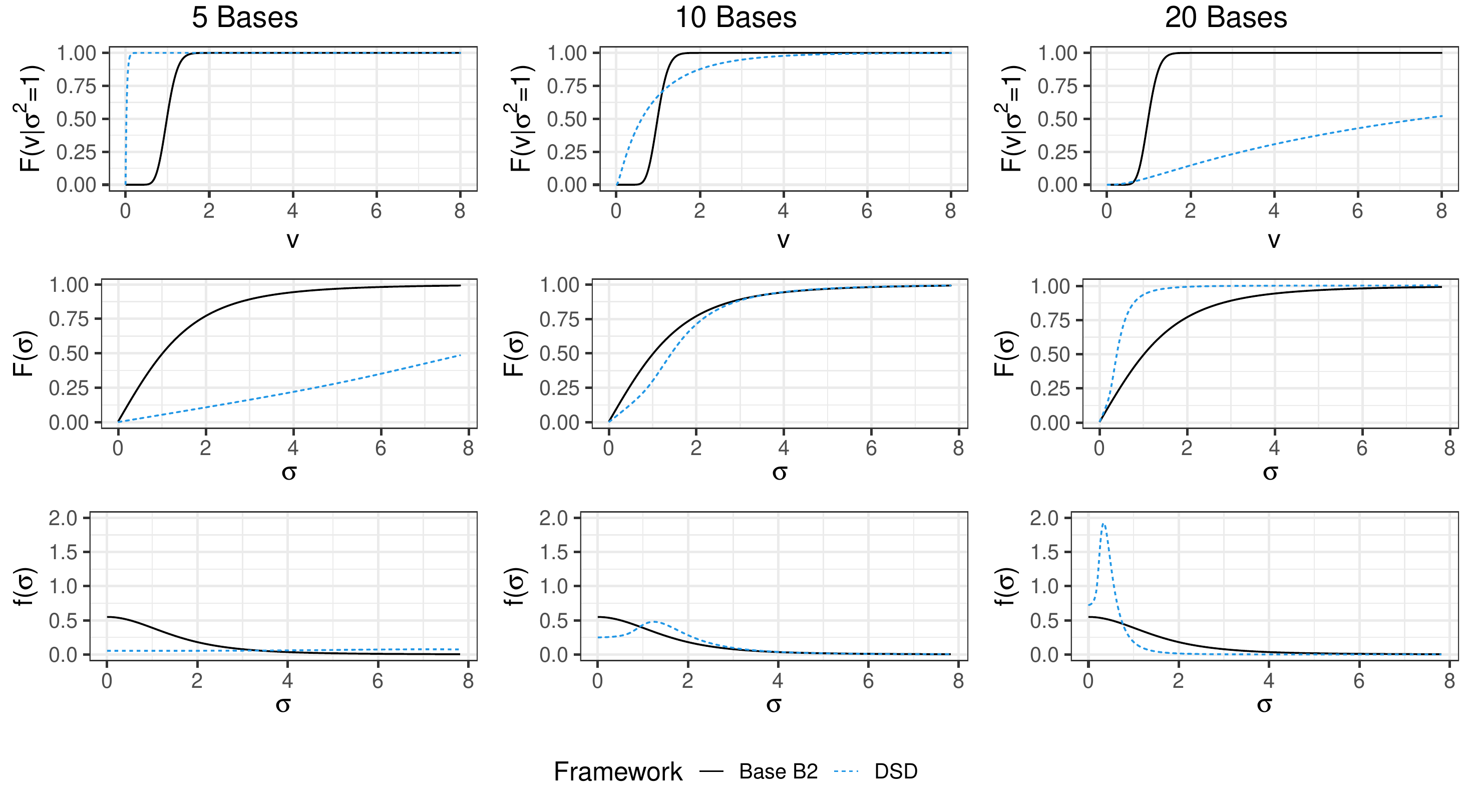}
 \caption{From top to bottom: CDFs of $V_\nu|\sigma^2=1$ (dashed line) and $\mathcal{V}|\sigma^2=1$ (solid line). CDFs and PDFs of implied priors on $\sigma$ with a B2 prior or a DSD prior for $\sigma^2$.}
 \label{prior_sim}
 \end{figure}

We consider three different values for $\sigma_y^2$, to control the signal-to-noise ratio $\rho=\frac{\text{Var}\left[\boldsymbol{\eta}\right]}{\text{Var}\left[\boldsymbol{\eta}\right]+\sigma_y^2}$, which is set equal to $\rho\in\{0.25,0.5,0.75\}$. Moreover, we set a grid of values for $m$, ranging from 5 to 30 spaced by 5. The same LGM is fitted on $B=500$ generated datasets, with different hyperpriors for the scale parameter $\sigma^2_\gamma$, distinguishing between standard priors (Half-Cauchy, the Half-$t$ on the standard deviation, Inv-Gamma($10^{-3}$,$10^{-3}$), labeled with IG-J, and Inv-Gamma(1,$5\times 10^{-5}$), labeled with IG-INLA, on the variance) and DSD priors for which hyperparameter values $q\in \{0.5,1.5,15\}$ and $\pi_0\in \{0.1, 0.25, 0.5, 0.75\}$ are explored. 

Since a Gaussian model is assumed for the data, a prior for the residual variance must be set. In line with the indications provided within the \texttt{rstanarm} package \citep{rstanarm}, we specified an exponential distribution on the residual standard deviation, with rate parameter $1/s_y^2$.
%\begin{itemize}
%    \item standard priors that are used in the literature, considering Half-Cauchy and the Half-$t$ on the standard deviation (with unitary scale parameter) and both Inv-Gamma(0.001,0.001) and Inv-Gamma(1,0.00005) on the variance, i.e. the approximation of the Jeffreys prior and the INLA default prior;
%    \item DSD priors, exploring different values of hyperparameter $q\in \{0.5,1.5,15\}$ and of $\pi_0\in \{0.1, 0.25, 0.5, 0.75\}$.
%\end{itemize}
%\textcolor{red}{The performances of point (posterior mean) and interval ($95\%$ credible intervals) estimators for the curve are monitored. As summary measures, the root mean squared error averaged over the curve (ARMSE) and the frequentist coverage are computed. Furthermore, to have a measure of the Monte Carlo uncertainty, also the MC sampling distribution of average error (AE) is reported.}
% $$
% ARMSE_{MC}^{(b)}=\sqrt{\frac{1}{50}\sum_{i=1}^{50}\left(\hat{\theta}_i^{(b)}-\theta_i^{(b)}\right)}
% $$
% $$
% ARMSE=\frac{1}{B}\sum_{b=1}^BARMSE_{MC}^{(b)}
% $$

%\begin{figure}[]
%\centering
%\includegraphics[width=0.9\linewidth]{V_cond}
% \caption{CDFs of the distribution of $V_\nu|\sigma^2=1$ with different number of bases. In bold, the CDF of $\mathcal{V}|\sigma^2=1$ is reported.}
% \label{V_cond}
% \end{figure}

 \begin{figure}[]
\centering
\includegraphics[width=\linewidth]{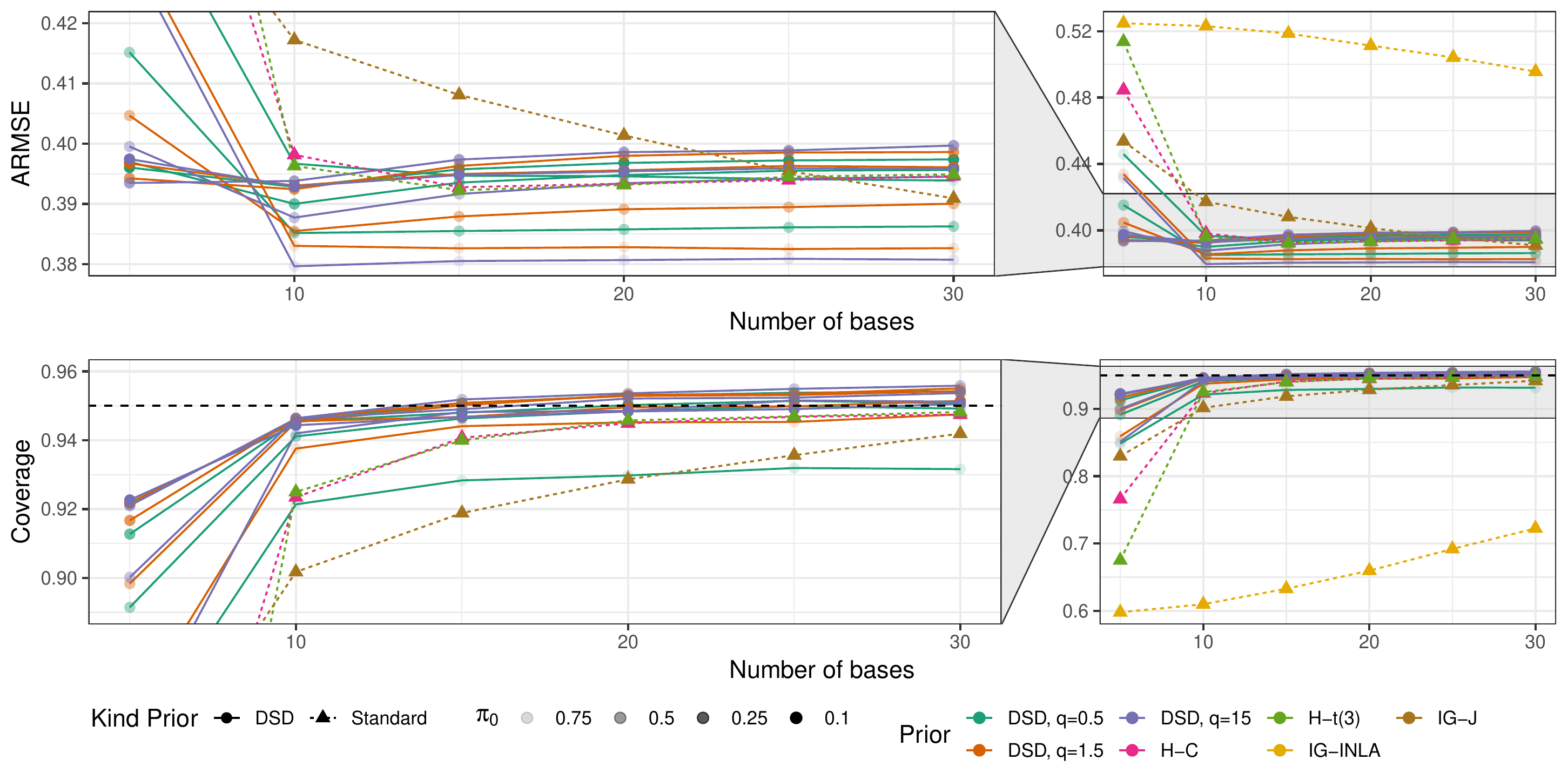}
 \caption{ARMSE and frequentist coverage with respect to the number of bases $m$ for the estimates under the considered priors and $\rho=0.25$. }
 \label{lines_rmse}
 \end{figure}

%Before moving to the analysis of the frequentist properties registered in the simulation study, a comparison of the DSD priors that are retrieved in different design settings can be useful to further clarify our approach. 

Once the prior elicitation step described in Section \ref{sec:elicit} is carried out, and the hyperparameters $b$, $p$, and $q$ are set, the goal is to impose a $_2\mathcal{F}_0(\alpha,\beta/b,p,q)$ prior distribution on $V_\nu$. We recall that it coincides with the marginal prior of $V_\nu$ when an i.i.d effect is considered, with a $B2(b,p,q)$ hyperprior on the scale. However, different model specifications imply different effects structures, that deviate from the i.i.d. structure. Hence, the DSD prior modifies the base B2 distribution to preserve the $_2\mathcal{F}_0(\alpha,\beta/b,p,q)$ marginal prior on $V_\nu$.

\begin{figure}[]
\centering
\includegraphics[width=\linewidth]{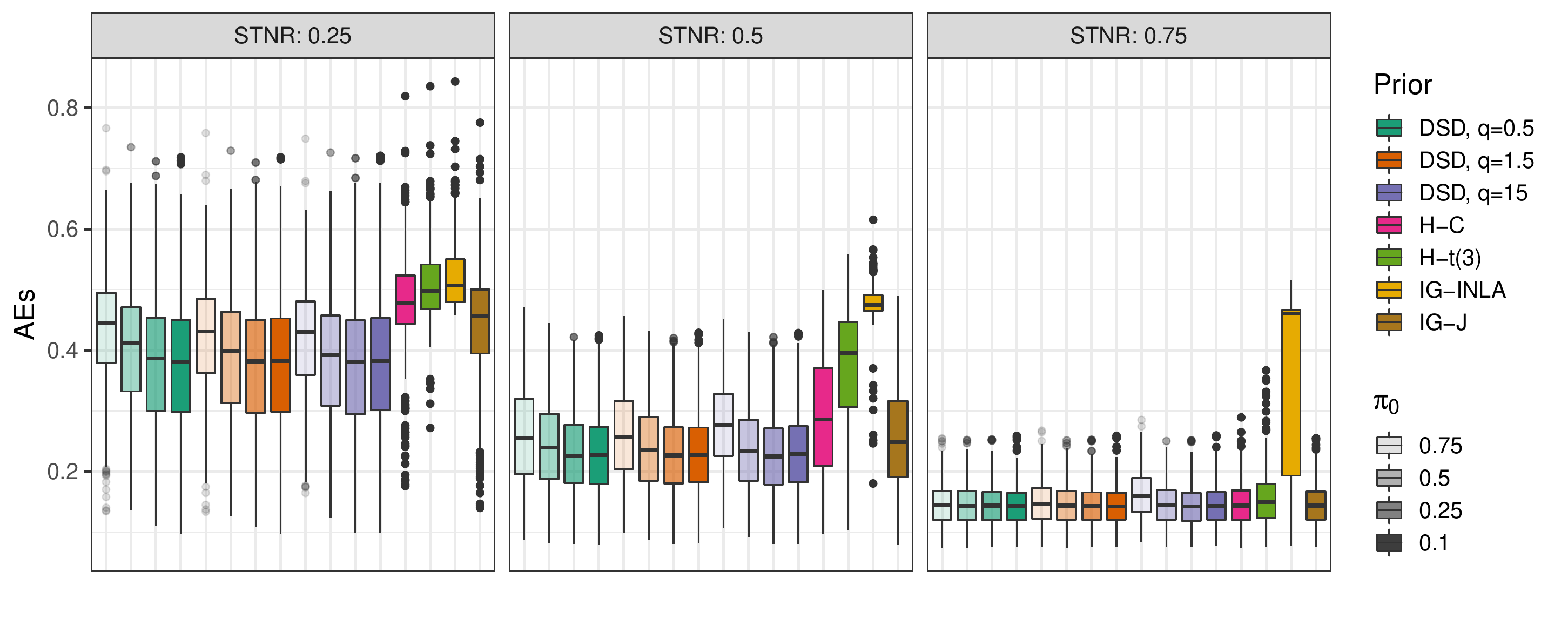}
 \caption{Monte Carlo distributions of AEs when $m=5$.}
 \label{box_rmse}
 \end{figure}
 
Figure \ref{prior_sim} allows noticing how the DSD prior changes with respect to the selected number of bases $m$: the base B2 distribution is plotted up to quantile 0.999 determining the range of the $x$ axis, and setting $\pi_0=0.5$ and $q=1.5$. We report the implied prior distribution on the standard deviation that has a positive finite density in 0 since $p=0.5$. 
As it can be noticed from the first row of Figure~\ref{prior_sim}, the CDF of $V_\nu|\sigma^2$ is steeper than the CDF of $\mathcal{V}|\sigma^2$ when $m=5$: as a consequence, a prior distribution flatter than the B2 density (see third row) is required on $\sigma^2$ to obtain the same marginal prior of the i.i.d. model. Increasing $m$ reduces the steepness of the CDF of $V_\nu|\sigma^2$: when $m=20$ the DSD prior induced on $\sigma$ appears to be peaked nearby 0, with the CDF that reaches the asymptote in 1 faster than the $B2$ prior.

In Figure \ref{lines_rmse}, the root mean squared error averaged over the curve (ARMSE) and the frequentist coverage are reported. The main differences among the priors involved in the study arise when $m=5$. In this case, the calibration carried out by DSD priors produces estimates characterized by a lower ARMSE and credible intervals with frequentist coverages closer to the nominal level. When $m$ increases, models performances tend to converge quickly, with the exception of those obtained with the IG-INLA prior. Figure \ref{lines_rmse} focuses on the case $\rho=0.25$, results for the remaining cases are reported in Figure \ref{supp_sim1}, where similar patterns can be detected.

\begin{figure}[h]
	\centering
	\includegraphics[width=\linewidth]{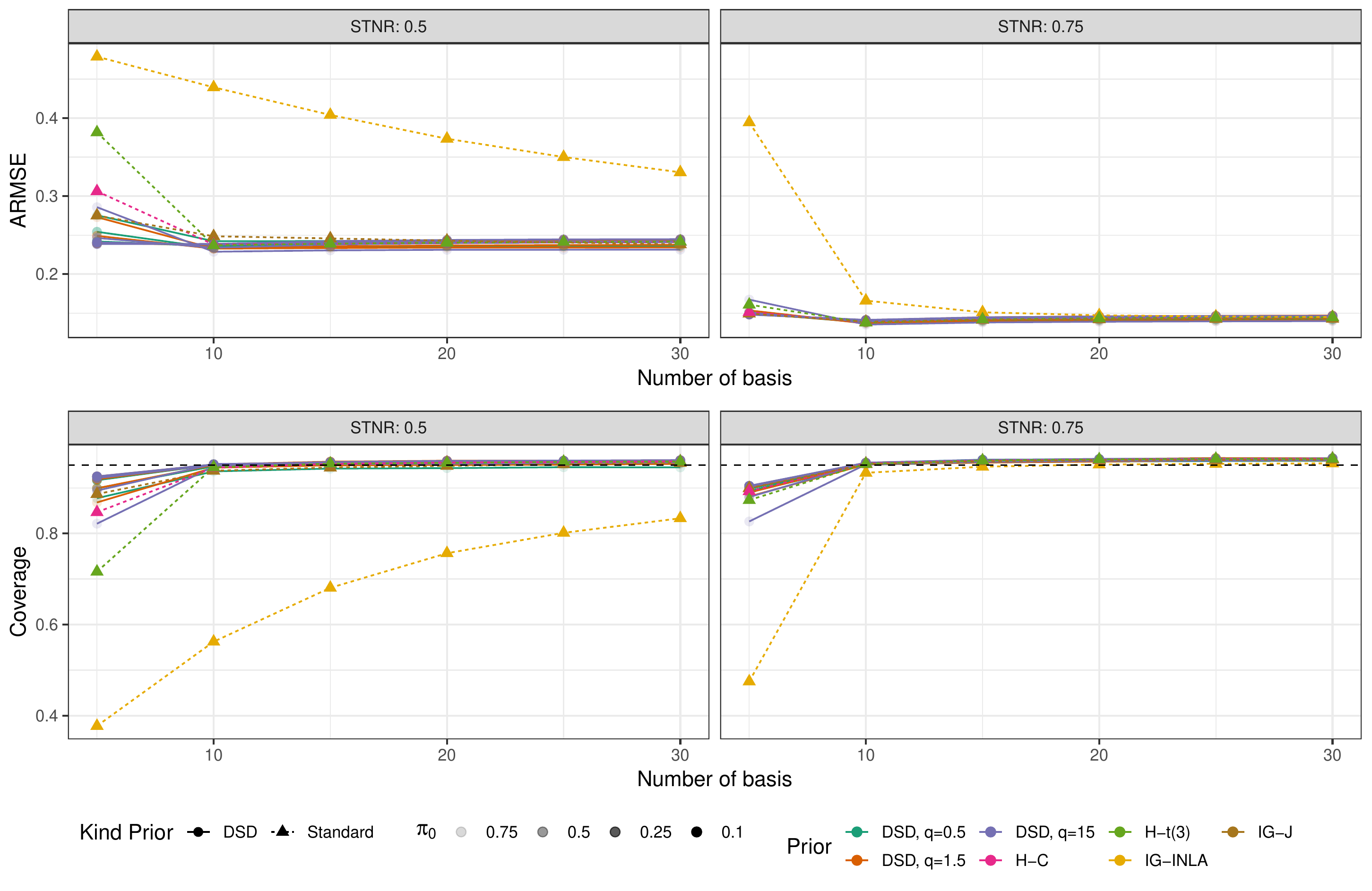}
	\caption{ARMSE and frequentist coverage with respect to the number of bases $m$ for the estimates under the considered priors.}
	\label{supp_sim1}
\end{figure}

 To provide an idea of the uncertainty of the sampling distribution, the distributions of the average errors (AEs) of the curves are depicted in Figure \ref{box_rmse} through boxplots. When $\rho\in\{0.25,0.5\}$, DSD priors deliver lower AEs than standard priors. Furthermore, with the exception of $\pi_0=0.75$ for which higher errors are obtained, the AEs distributions obtained with other values of $\pi_0$ are fairly stable. Lastly, results obtained for $\rho=0.75$ show little variability, with the exception of the IG-INLA prior that produces markedly higher AEs values.

Lastly, the three panels of Figure \ref{box_rmse} highlight that, when using DSD priors, changes in posterior inference are mainly due to $\pi_0$, while results are less reactive to changes of the $q$ values.  

	\section{Applications}\label{sec:application}
	In this section, we show two applications based on datasets available in the \texttt{INLA} package and analyzed in \cite{rue2005gaussian} among others: the Munich rental and the Tokyo rainfall datasets. %The last application concerns the cherry blossom dataset \citep{aono2010clarifying}. 

\subsection{Munich Rental Data}

\begin{table}[]
	\begin{center}
\begin{tabular}{lccccc}
\toprule
& $m_j$ &$\widetilde\alpha$ & $\widetilde\beta$ & $\mathbb{E}[V_{\nu_j}|\sigma^2_{\nu_j}=1]$  & $\mathbb{V}[V_{\nu_j}|\sigma^2_{\nu_j}=1]$\\
\midrule
     size & 134 &0.735 & $1.4\times 10^{-4}$ & 5001.1 &  $3.4\times 10^7$\\
     year & 84 &0.602 & $2.7\times 10^{-4}$ & 2190.9 & $7.9\times 10^6$\\
     location & 380 &9.375 & 11.581 & 0.809 & 0.0698 \\
     \bottomrule
\end{tabular}
\caption{Characteristics of the model components sampling variances.}
\label{tab:munich-abtil}
\end{center}
\end{table}

\begin{figure}
    \centering
    \includegraphics[width=\linewidth]{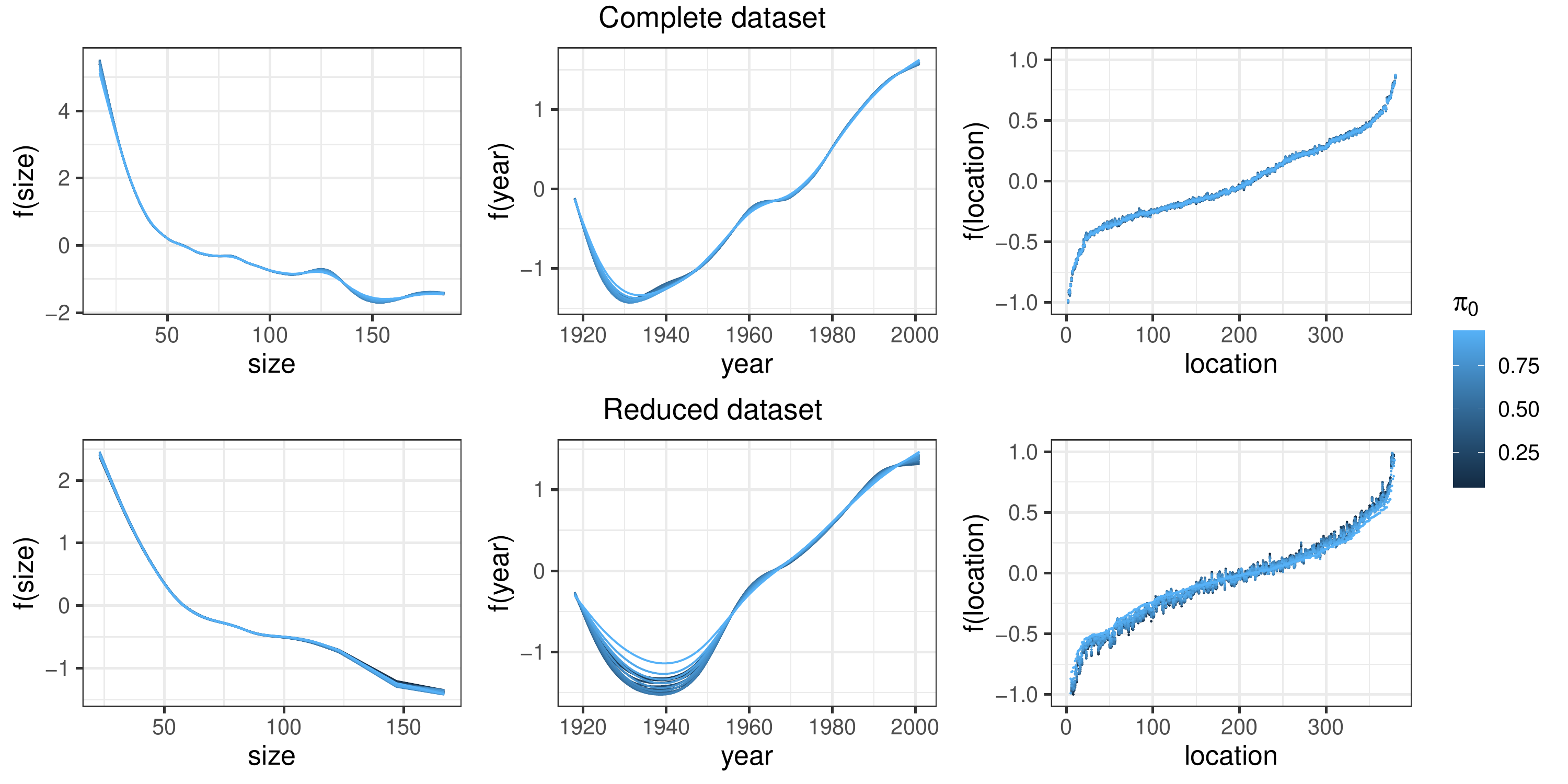}
    \caption{$\mathbb{E}[f(\mathbf x_j)|\mathbf{y}]$ of model components for complete and reduced datasets (q=1.5).}
    \label{fig:munich_eff}
\end{figure}

The dataset comprises $n=2035$ observations from the 2003 Munich rental guide. The response variable is the rent price (per square meter in Eur\\
s); covariates comprise spatial location, floor space, year of construction, and a set of dichotomous variables describing house features such as presence/absence of central heating, bathroom, etc. To study prior sensitivity, we also consider a reduced dataset obtained by randomly sampling $n=300$ observations. Following previous analyses, a Gaussian likelihood is adopted and the linear predictor is specified as
\[
\boldsymbol{\eta}=\mathbf{1}\beta_0+\mathbf{X}\boldsymbol{\beta}+\boldsymbol{\nu}_{\text{year}}+\boldsymbol{\nu}_{\text{size}}+\boldsymbol{\nu}_{\text{loc}},
\]
where
\[
\boldsymbol{\nu_j}=\mathbf{Z}_j\boldsymbol{\gamma}_j;\quad\boldsymbol{\gamma}_j|\sigma^2_{\gamma_j}\sim\mathcal{N}\left(\mathbf{0},\sigma^2_{\gamma_j}\mathbf{K}_{\gamma_j}\right),\quad j\in\{\text{year, size, location}\}.
\]
Both $\mathbf{K}_{\gamma_\text{year}}$ and $\mathbf{K}_{\gamma_\text{size}}$ are precision matrices of a continuous time second-order RW, hence their rank-deficiency is $\kappa=2$. Since their null spaces are spanned by a first-order polynomial, corresponding linear terms are added to the linear predictor.  $\mathbf{K}_{\gamma_\text{loc}}$ is built as $\mathbf{D}-\mathbf{W}$ where $\mathbf{W}$ is the adjacency matrix of Munich districts and $\mathbf{D}$ is diagonal with $i$-th diagonal entry corresponding to the number of neighbors to the $i$-th district: being the graph connected, $\kappa=1$ and no terms need to be added to the linear predictor. 
All random effects are constrained as described in Section \ref{sec:IGMRF}. 

\begin{figure}
	\centering
	\includegraphics[width=\linewidth]{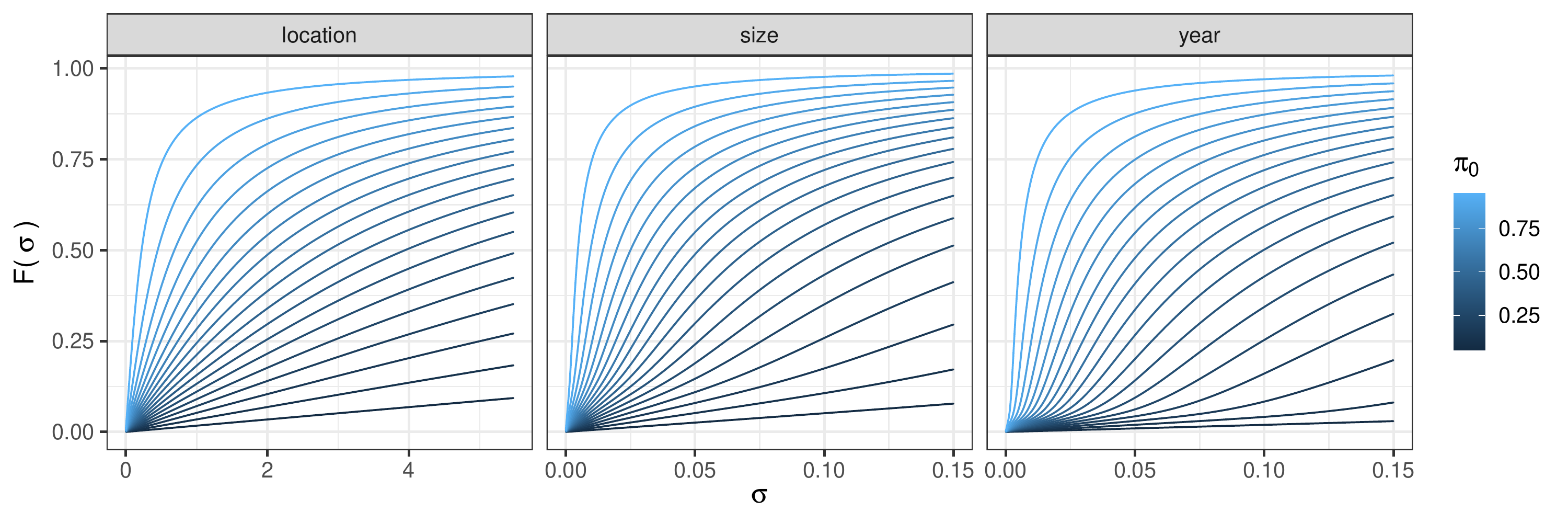}
	\caption{CDFs of the priors on $\sigma_{\gamma_j}$ for the three random effects under the complete datasets.}
	\label{fig:dsdprior}
\end{figure}

Table \ref{tab:munich-abtil} shows parameters of the conditional distributions $V_{\nu_j}|\sigma^2_{\gamma_j}$, where the last two columns report expectation and variance for $\sigma^2_{\gamma_j}=1$. The same prior on scale parameters would result in sensibly different marginal distributions of $V_{\nu_j}$, with random effects built on RW structures being largely more variable a priori than the spatial random effect. The CDFs of DSD priors as a function of $\pi_0$ are shown in Figure \ref{fig:dsdprior}: focusing on the ranges of the $x$-axes, one can note that probability mass is distributed on far larger values for the spatial effect with respect to other components, in order to compensate for the aforementioned features of the conditional sampling variances. 

\begin{figure}
    \centering
    \includegraphics[width=\linewidth]{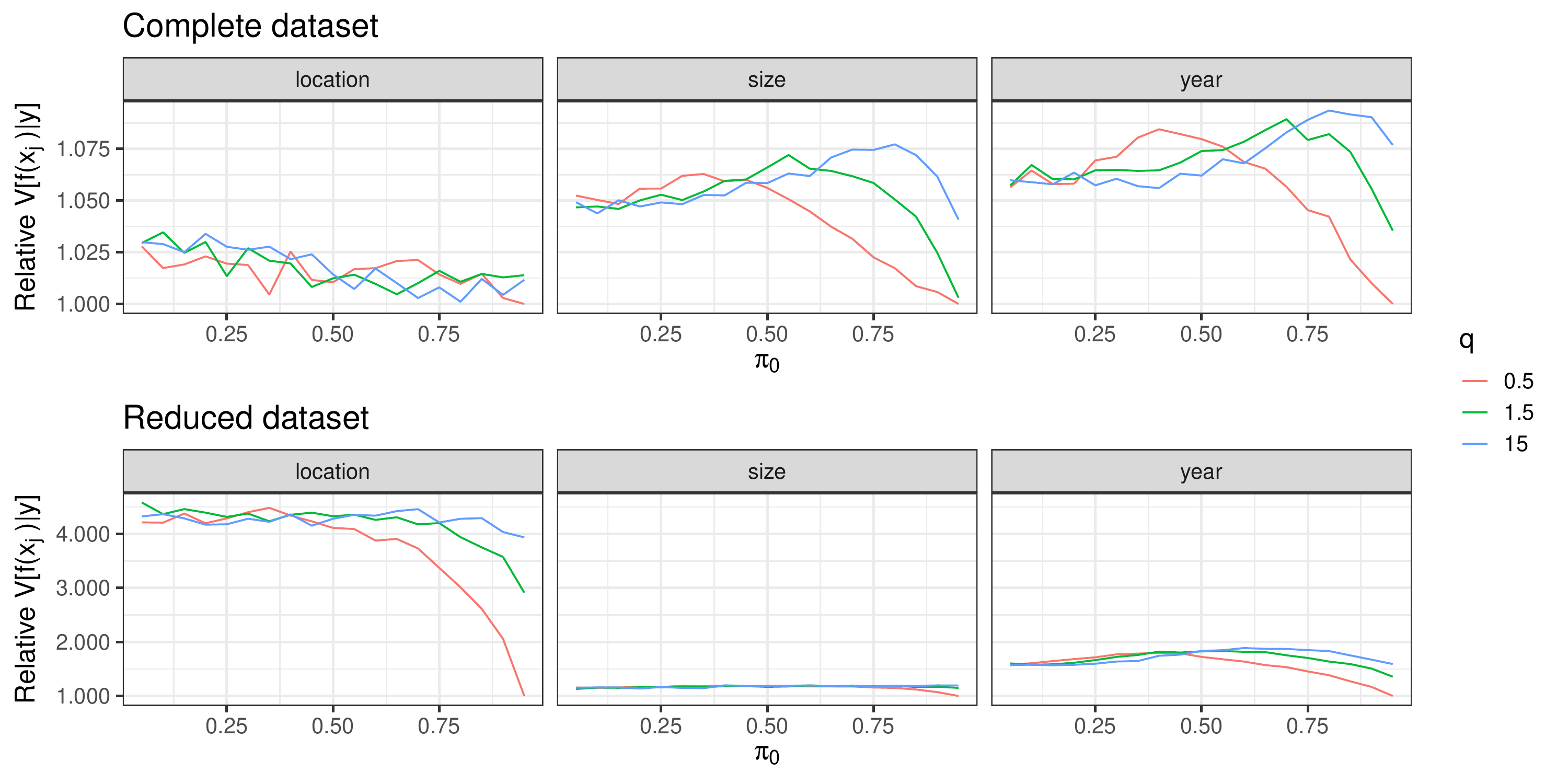}
    \caption{$\mathbb{V}[\mathbb{E}[f(j)|\mathbf{y}]]$ divided by the component-wise maximum as a function of $\pi_0$.}
    \label{fig:munich_vsplit}
\end{figure}

In Figure \ref{fig:munich_eff}, the posterior means of effects of size, year and location are reported for $\pi_0$ ranging from 0.05 to 0.95 and $q=1.5$. In the complete dataset, prior sensitivity is limited, particularly in comparison with the reduced dataset. Furthermore, the spatial effect shows the most noticeable reaction to the loss of information. To highlight the changes due to $\pi_0$, sensitivity curves are reported in Figure \ref{fig:munich_vsplit}, where relative  $\mathbb{V}[\mathbb{E}[f(j)|\mathbf{y}]]$ are shown for the three predictor effects $f(j)$. All variances are divided by the component-wise smallest variance in order to emphasize relative changes. As expected, increasing $\pi_0$ enhances smoothing: location effect is the least sensitive to prior specification in the complete dataset, while the year effect shows some sensitivity in both cases.

Lastly, we study the decomposition of the linear predictor posterior variance $V_\eta|\mathbf{y}$, joining fixed effects. Posterior covariances $C_{\nu_j,\nu_k}=\frac{\boldsymbol{\nu_j}^\top\mathbf{M}\boldsymbol{\nu_k}}{n-1}\bigg|\mathbf{y}$ are negligible, hence posterior variances $V_{\nu_j}|\mathbf{y}$, whose distributions are shown in Figure \ref{fig:munich_vf}, give clear-cut information on the contribution of each model component to the explained variability which receives the same prior weight by means of DSD priors. The size variable turns out to be the most important covariate for explaining rental prices, followed by fixed effects, year of construction and location.

\begin{figure}
	\centering
	\includegraphics[width=.8\linewidth]{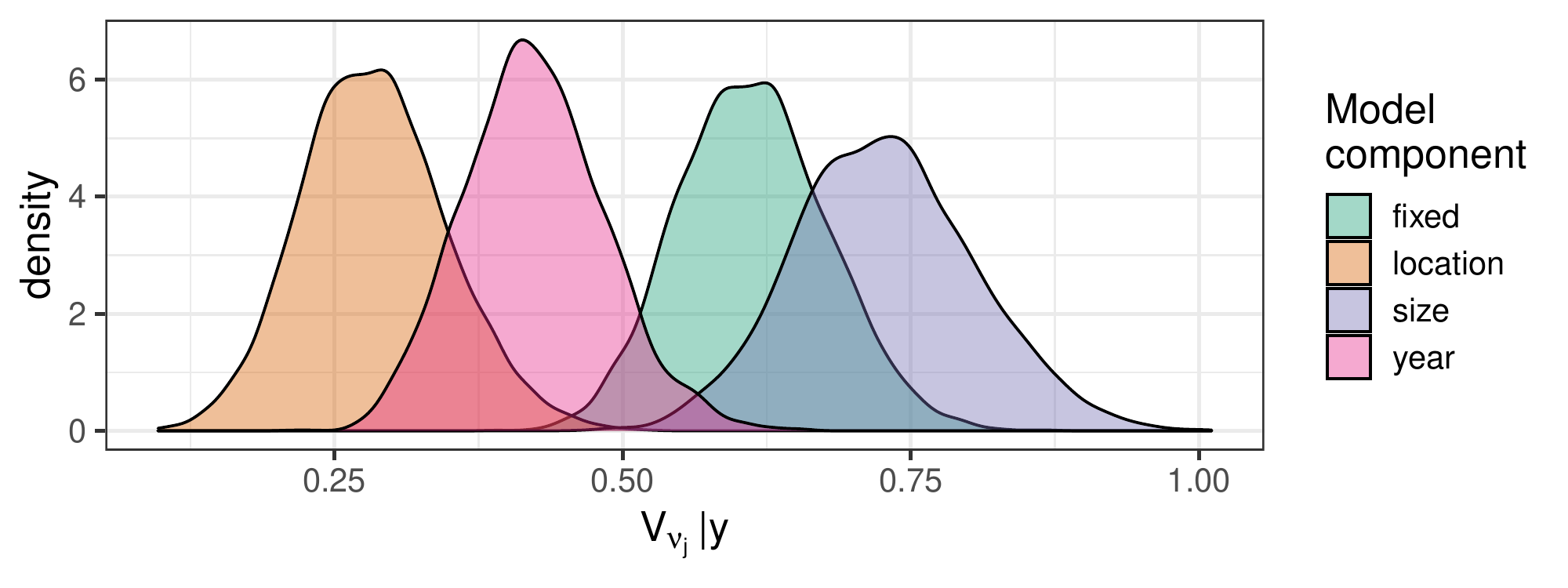}
	\caption{Posterior distributions $V_{\nu_j}|\mathbf{y}$ under DSD priors with $q=1.5$ and $\pi_0=0.9$.}
	\label{fig:munich_vf}
\end{figure}

 \begin{figure}
	\centering
	\includegraphics[width=\linewidth]{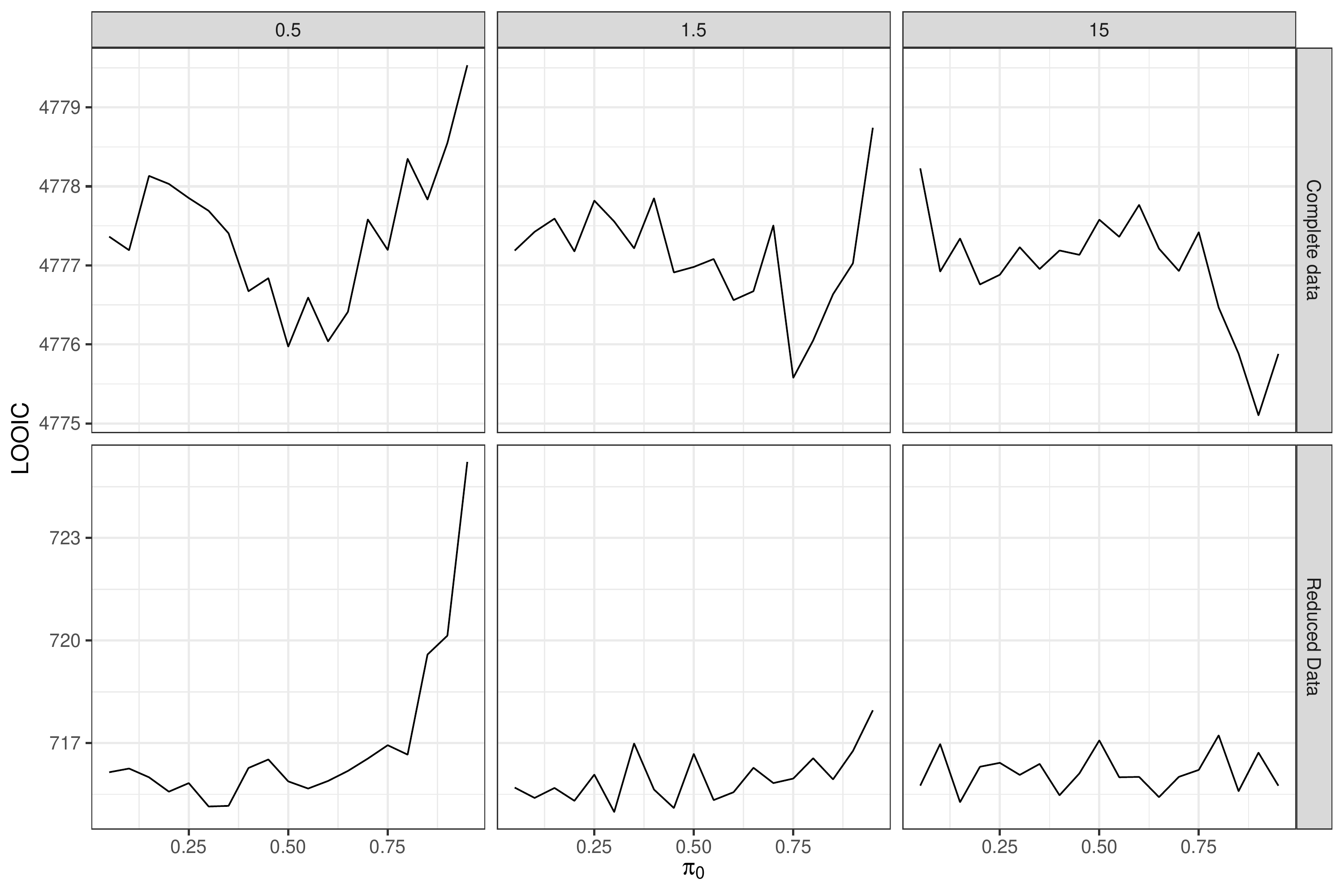}
	\caption{LOOIC related to model fitted on complete and reduced datasets (q=1.5).}
	\label{fig:munich_loos}
\end{figure}

\subsection{Tokyo Rainfall Data}

The dataset comprises $n=731$ daily dichotomous observations which equal 1 if more than 1 mm of rainfall was recorded in Tokyo during 1983 and 1984, and 0 otherwise. The aim of the application is to estimate the rainfall probability $p_t$ on a calendar day $t=1,\ldots,366$. The Binomial likelihood is
\[
y_{t}|p_t\sim\text{Bin}(p_t, n_{t}),\qquad t=1,\ldots,366,
\]
where $n_t = 2$ for $t\neq 60$ and $n_{60}=1$. The linear predictor is 
$
g(\mathbf{p})=\mathbf{1}\beta_0+\boldsymbol{\nu_{\text{day}}}
$
where
\[
\boldsymbol{\nu_{\text{day}}}=\mathbf{Z}\boldsymbol{\gamma}_{\text{day}};\quad\boldsymbol{\gamma}_{\text{day}}|\sigma^2_{\gamma_{\text{day}}}\sim\mathcal{N}\left(\mathbf{0},\sigma^2_{\gamma_{\text{day}}}\mathbf{K}_{\gamma_{\text{day}}}\right).
\]
In this case, the random effect design matrix is $\mathbf{Z}=\mathbf{I}_{366}$ and the structure 
$\mathbf{K}_{\gamma}$ corresponds to a circular RW of order two, to introduce conditional dependence between the first and last day of a year \citep{rue2005gaussian}. Note that the IGMRF order is $\kappa=1$, so the model needs only one linear constraint and no term needs to be added to the linear predictor.

\begin{figure}
    \centering
    \includegraphics[width=\linewidth]{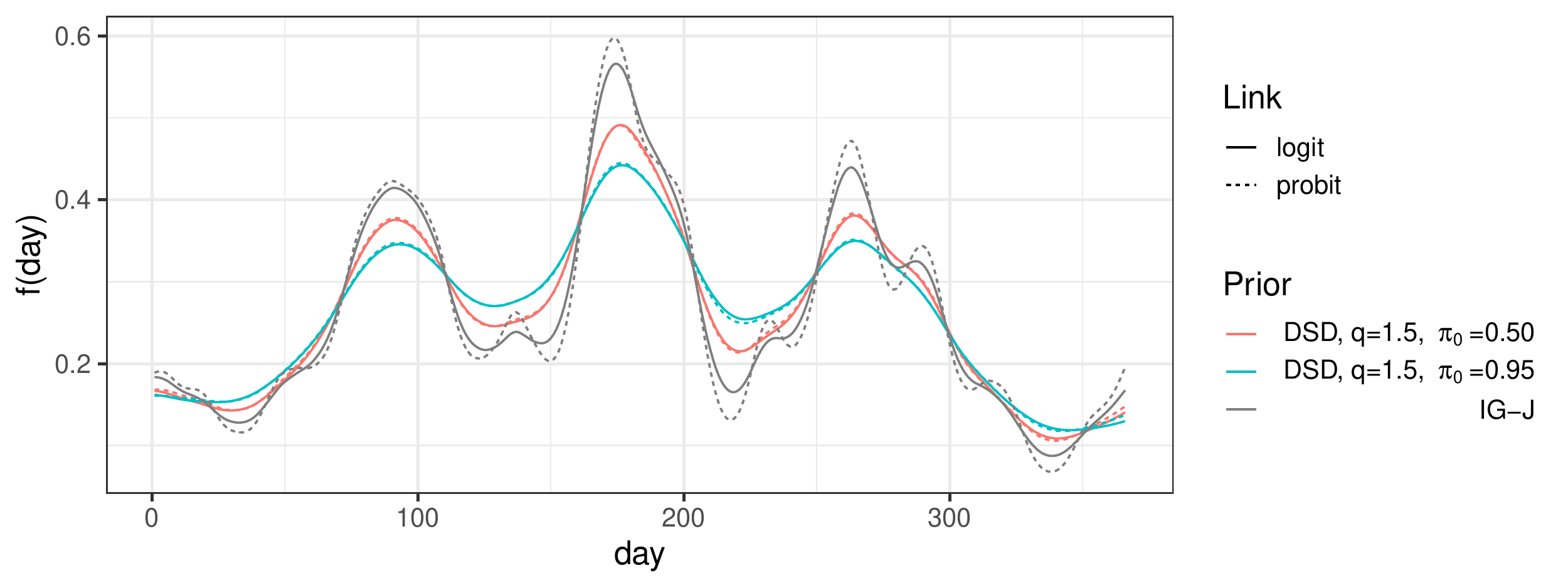}
    \caption{Fitted probabilities with respect to calendar days.}
    \label{fig:pred_tokyo}
\end{figure}

We estimate the model adopting both the logit and probit link. In the former case, we set $c$ in \eqref{eq:pi0} as the pseudo-variance
$c=\bar y^{-1}(1-\bar y)^{-1}=5.16$ (see Table 1 in \cite{piironen2017}), while $c=\bar y(1-\bar y)/\phi(\Phi^{-1}(\bar y))^2=1.82$ for the probit link. This leads to $b=26.5$ and $b=9.34$ for the logit and probit links respectively when $q=1.5$ and $p=0.5$ as suggested in Section \ref{sec:elicit} and $\pi_0=0.5$.

The merit of assigning different priors in the probit and the logit scale is highlighted in Figure \ref{fig:pred_tokyo}, where one can see that adopting the same IG-J prior independently on the link delivers a sensibly different smoothness of the fitted probabilities as a function of calendar day. On the other hand, DSD priors show limited variation with respect to the type of link function, being the smoothness of the predicted probabilities mainly affected by $\pi_0$. As expected, in this application, posterior inference shows a marked sensitivity to $\pi_0$, as shown in the sensitivity curves reported in the left panel of Figure \ref{fig:loo_tokyo}.

In the right panel of Figure \ref{fig:loo_tokyo}  the leave-one-out cross-validation information criterion \citep[LOOIC,][]{vehtari2017practical} is reported: it can be seen that LOOIC increases as the amount of shrinkage (determined by $\pi_0$) increases. We observed a similar behavior in the Munich rental data application (see Figure \ref{fig:munich_loos}) in the reduced dataset, suggesting that extreme values of $\pi_0$ can give rise to excessive smoothing resulting in worst model performances. On the basis of these considerations, we suggest to set $\pi_0=0.5$, completing the default setting sketched in Section \ref{sec:elicit}.

\begin{figure}
    \centering
    \includegraphics[width=.9\linewidth]{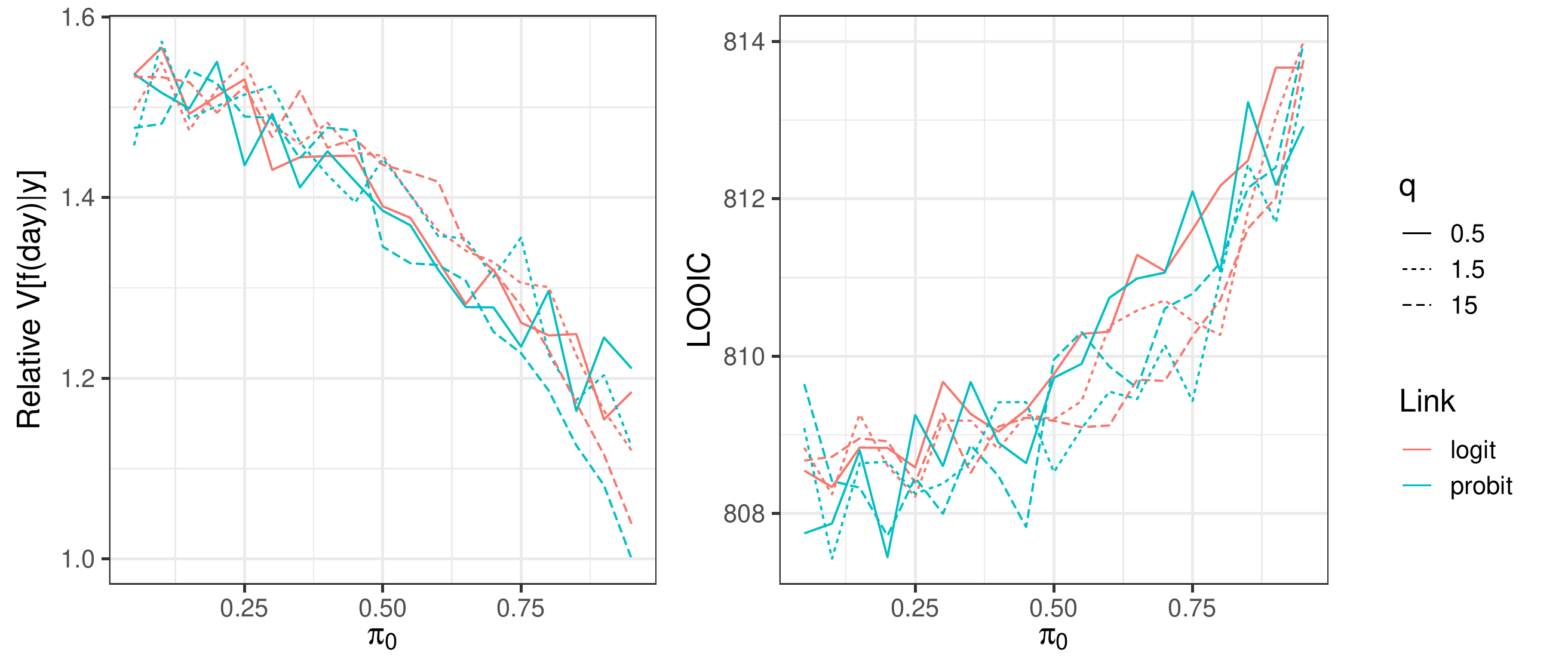}
    \caption{$\mathbb{V}[\mathbb{E}[f(j)|\mathbf{y}]]$ divided by the maximum as a function of $\pi_0$ (left panel). Leave One Out Information Criterion as a function of $\pi_0$ (right panel).}
    \label{fig:loo_tokyo}
\end{figure}

	\section{Concluding Remarks}\label{sec:conclusion}
	In this paper, a unified strategy for prior specification of scale parameters in Latent Gaussian Models has been proposed. This is a debated problem in the literature about Bayesian Hierarchical models and several interesting proposals have been developed to date. In our opinion, the lack of agreement about sensible default strategies for prior specification in this kind of models is due to three aspects that hinder clear-cut interpretation of the priors on parameters that act as mere scaler: design matrices, structure matrices and linear constraints. From an applied point of view, we believe that the most important phase of model specification concerns indeed these three components which contain: auxiliary information on the phenomenon under study (design), prior beliefs on the correlation structure/smoothness of model components (structure) and precautions needed for model identifiability or, possibly, for customizing parameters interpretability (linear constraints). Scale parameters are far less relevant in model specification, but they require careful prior specification in order to obtain a sensible allocation of prior variability to model components. As shown in Section \ref{sec:sdpri}, this can be done by taking into account the whole model architecture, comprising the link function. The $\pi_0$ parameter in DSD priors can be interpreted as a measure of plausibility of observed data under the model: for this reason, we find it sensible to set $\pi_0$ on the basis of the data variability when a Gaussian likelihood is concerned or on pseudo-variance in generalized linear models.
Concerning usability and easiness of application of LGMs, we think that DSD priors can be a useful tool for practitioners to be used as a safe default prior, allowing for straightforward sensitivity analysis in disparate situations comprising spatial, temporal, spatio-temporal, ANOVA, semi-parametric models and in general all models that can be cast coherently with the models summarised in Table \ref{tab:intro}.

Besides the practical aspects, another goal of the paper is to deepen the theory behind LGMs, trying to characterize conditional and marginal sampling variances of linear predictor components. These quantities provide a meaningful interpretation of scale parameters, useful in the prior elicitation step. These developments allow contextualizing previous contributions in the field like those by \citet{Sorbye2014} and \citet{Klein2016}. Indeed, the scaling procedure by the former can be seen as an attempt of removing the impact of the structure in the conditional expectation of the sampling variance, as deducible by results in Section \ref{sec:distr_V}. On the other hand, finding a prior that solves the integral equation in \eqref{eq:inteqth} is in line with the proposal of scale dependent priors by \citet{Klein2016}, where the prior on scale parameters is retrieved numerically, relying on a different synthesis of the effect dispersion and incorporating the design through arbitrary covariate patterns. In this sense, taking the sampling variance of the effect as focal quantity leads us to meaningful distributional results by resorting to the theory of QFs.

	\appendix
	\section{Basics of Mellin Transform}\label{sec:mellin}
The Mellin transform, strictly related to the more famous Laplace and Fourier transforms, is a mathematical tool largely used throughout the paper. 
In what follows, given a function $g(\cdot)$, its Mellin transform  \citep{paris2001asymptotics,poularikas2018transforms} is denoted as $\widehat{g}(\cdot)$. For a random variable $X$ having the positive real axis as support, the Mellin transform of its density function is defined as:
$
\widehat{f}_X(z)=\int_0^{+\infty} x^{z-1}f_X(x)\mathrm{d}x,
$
where $z=h+iy \in \mathbb{C}$.
Conversely, the density function $f_X(\cdot)$ can be recovered from $\widehat{f}_X(\cdot)$ by the inverse Mellin transform:
\begin{equation}\label{eq:mellin_inv}
	f_X(x)=\frac{1}{2\pi i}\int_{h-i\infty}^{h+i\infty} x^{-z}\widehat{f}_X(z)\mathrm{d}z,
\end{equation}
where $h=\Re(z)$ individuates a Bromwich path of integration included in the strip of analyticity of $\widehat{f}_X(\cdot)$. Once the strip of analyticity is identified, absolute convergence of both integrals is guaranteed.

The Mellin transform is a valuable tool in solving integral equations having form:
$$
f_X(x)=\int_0^{+\infty} y^{-1}f_{X|Y}(x/y)f_Y(y)\mathrm{d}y,
$$
since exploiting the results contained in \citep[Supplement~8]{polyanin2008handbook}, the Mellin transform of $f_X(x)$ can be written as the product:
\begin{equation}\label{eq:mellin_rel}
	\widehat{f}_X(z)=\widehat{f}_{X|Y}(z)\widehat{f}_Y(z).
\end{equation}

\section{Quadratic forms distribution}\label{app:QF}
Let us consider a $p$-dimensional random vector distributed as a multivariate normal:
\begin{equation}
	\mathbf{x}\sim\mathcal{N}_p\left(\boldsymbol{0},\boldsymbol{\Sigma}\right).
\end{equation}
The canonical expression of a quadratic form in $\mathbf{x}$ is:
\begin{equation}\label{eq:can_QF}
	Q(\mathbf{x})=\mathbf{x}^T\mathbf{A}\mathbf{x}
\end{equation}
where $\mathbf{A}\in\mathbb{R}^{p\times p}$ is a symmetric matrix.
Recalling the decomposition $\bm \Sigma=\bm L^T\bm L$, the distribution of a quadratic form can be expressed also as a weighted sum of $r$ central chi-squared distribution having $1$ degree of freedom and positive weights $\lambda_i$, $i=1,\dots,r$, that are the eigenvalues of the matrix $\bm L^T \bm{A L}$. The number of sum components $r$ is determined by the rank of $\bm L^T \bm{A L}$:
\begin{equation}\label{eq:c_lin_QF}
	Q(\mathbf{x})\sim\sum_{i=1}^{r}\lambda_i\chi^2_1,
\end{equation}
whose density function can be written as an infinite sum \citep{ruben1962probability}:
\begin{equation*}
	f_{Q}(q)=\sum_{k=0}^{\infty} c_k\frac{q^{r/2+k-1}\exp\left\{-\frac{v}{2\rho} \right\}}{(2\rho)^{r/2+k}\Gamma(r/2+k)},
\end{equation*}
where the coefficients $c_k$ are defined recursively as follows:
\begin{equation*}
	\begin{aligned}
		c_0&=\prod_{i=1}^{r} \left(\frac{\rho}{\lambda_i} \right)^{\frac{1}{2}};\quad
		c_k=(2k)^{-1} \sum_{l=0}^{k-1}g_{k-l}c_l,\quad k\geq 1;\\
		g_k&=k\sum_{i=1}^r \left(1-\frac{\rho}{\lambda_i}\right)^k,\quad k\geq 1;
	\end{aligned}
\end{equation*}
and $\rho$ is a constant chosen to improve the convergence of the infinite sum and respect the condition: $\left|1-\frac{\rho}{\lambda_i} \right|<1,\  \forall i$.

	\bibliographystyle{agsm}
	\bibliography{biblio}

\end{document}